%% file: main.tex
\pdfoutput=1
\documentclass[acmsmall]{acmart}
\usepackage{paralist}
\usepackage{graphicx}
\usepackage{hyperref}
\usepackage{amsmath}
\usepackage{cleveref}
\usepackage{xcolor}
\usepackage{comment}
\usepackage{dsfont}
\usepackage{adjustbox}
\usepackage[inline]{enumitem}

\usepackage{listings}
\usepackage{xspace}
\usepackage{cancel}

\usepackage{macros}
\usepackage{lipsum}
\usepackage{comment}

\usepackage{subfig}

\usepackage{tikz}
\usetikzlibrary{arrows, positioning, fit, plotmarks}
\usepackage{tikzsymbols}
\usepackage{pgfplots}
\pgfplotsset{compat=1.16}

\crefname{equation}{}{}

\Crefname{definition}{Def.}{Def.}
\Crefname{figure}{Fig.}{Fig.}
\crefname{figure}{Fig.}{Fig.}
\Crefname{theorem}{Thm.}{Thm.}
\crefname{theorem}{Thm.}{Thm.}
\Crefname{section}{Sec.}{Sec.}

\acmJournal{PACMPL}
\acmVolume{1}
\acmNumber{CONF} 
\acmArticle{1}
\acmYear{2018}
\acmMonth{1}
\acmDOI{} 
\startPage{1}

\setcopyright{none}

\title{Hyperproperty Verification 
as CHC Satisfiability}
\author{Shachar Itzhaky}
\affiliation{
  \institution{The Technion}
  \country{Israel}
}
\author{Sharon Shoham}
\affiliation{
  \institution{Tel-Aviv University}
  \country{Israel}
}
\author{Yakir Vizel}
\affiliation{
  \institution{The Technion}
  \country{Israel}
}

\begin{document}

\input{00-abstract}

\maketitle

\input{01-intro}

\input{02-overview}
\input{03-background}
\input{04-transformation}

\input{05-safety}

\input{06-beyond}

\input{07-evaluation}
\input{08-related}
\input{09-conclusion}

\subsubsection*{Acknowledgment} The research leading to these results has received funding from the
European Research Council under the European Union's Horizon 2020 research and innovation programme (grant agreement No [759102-SVIS]). This research was partially supported by the Israeli Science Foundation (ISF) grant No. 2875/21 and No. 2117/23,
and by the NSF-BSF grant No. 2018675.





\bibliography{refs,refs2}

\clearpage
\appendix
\input{B-appendix-comparators}

\end{document}

%% file: 00-abstract.tex
\begin{abstract}
Hyperproperties specify the behavior of a system or systems across multiple executions, and are being
recognized as an important extension of regular temporal properties. 
So far, such properties have resisted comprehensive treatment by modern software model-checking approaches such as IC3/PDR,
due to the need to find not only a (relational) inductive invariant but also a \emph{total} alignment of different executions that facilitates simpler inductive invariants.

We show how this treatment is achieved via a reduction
from the verification problem of  (a certain class of) $\forall^*\exists^*$ hyperproperties to Constrained Horn Clauses (CHCs). Our starting point is a set of universally quantified formulas in first order logic (modulo theories) that encode the verification of $\forall^*\exists^*$ hyperproperties over infinite-state transition systems. The first-order encoding uses uninterpreted predicates to capture the (1) witness function for existential quantification over traces, (2) alignment of executions, and (3) corresponding inductive invariant. The encoding also uses first-order theories (e.g., arithmetic and arrays) for modeling the transition system.
Such an encoding was previously proposed for $k$-safety properties, and its extension to (a class of) $\forall^*\exists^*$ properties is straightforward.
Unfortunately, finding a satisfying model for the resulting first-order formulas is beyond reach for modern first-order satisfiability solvers. In general, state-of-the-art solvers struggle when quantifiers are combined with theories and uninterpreted predicates. A notable exception is the CHC fragment of first-order logic. The unique structure of CHCs makes it possible to adopt software model checking techniques for solving them. Indeed, CHC solving is a fruitful research area, supported by effective automation.
Unfortunately, the natural encoding of hyperproperty verification problems is not Horn.
Previous works tackled this obstacle by developing specialized techniques for solving the aforementioned first-order formulas. In contrast, we show that the same problems can be encoded as CHCs and solved by existing CHC solvers, benefiting from the ongoing development of these already successful tools.

Our key technical contribution is a logical transformation of the aforementioned sets of first-order formulas to equi-satisfiable sets of CHCs. The transformation to CHCs is sound and complete, and applying it to the first-order formulas that encode verification of 
hyperproperties leads to a CHC encoding of these problems.
We implemented the CHC encoding in a prototype tool and show that, using existing CHC solvers for solving the CHCs, the approach already outperforms state-of-the-art tools for hyperproperty verification by orders of magnitude.


\end{abstract} 

%% file: 01-intro.tex
\section{Introduction}
\label{intro}

Hyperproperties~\cite{DBLP:journals/jcs/ClarksonS10} are properties that relate multiple execution traces, either taken from a single program or from multiple programs. Checking such properties is known as \emph{relational verification}, and is essential when
reasoning about security policies, program equivalence, concurrency protocols, etc. 
Existing specification languages for hyperproperties~\cite{DBLP:conf/post/ClarksonFKMRS14,DBLP:conf/cav/BaumeisterCBFS21,DBLP:conf/pldi/SousaD16} extend standard ones, e.g., temporal logic or Hoare logic, with (explicit or implicit) quantification over traces.
This shifts the focus from properties of individual traces to properties of \emph{sets} of traces.
For example, 
$k$-safety~\cite{DBLP:journals/jcs/ClarksonS10} is a class of hyperproperties, where $k$ universal quantifiers are used to define a relational invariant over
states originating from $k$ 
traces.

This paper addresses verification of hyperproperties with $\forall^* \exists^*$
quantification over traces and a body of the form $\square\phi$ (where $\square$ stands for ``globally'').
This fragment captures many hypersafety (e.g., the aforementioned $k$-safety)
and hyperliveness properties,
and was 
shown by \cite{DBLP:conf/cav/BeutnerF22} to express a wide class of
properties of interest, including generalized non-interference (GNI)~\cite{DBLP:conf/sp/McCullough88}.

Verification of hyperproperties
is more challenging than
verification of single-trace properties, and, as a result, has
gained a lot of attention in recent years.
Unlike single-trace properties, verification of properties of $k$ traces requires the discovery of \emph{relational} inductive invariants, which define the
relation between states of $k$ execution traces.
Since the construction of invariants that hold between \emph{any} $k$ reachable states is hard (or even impossible, depending on the assertion logic),
proving hyperproperties often hinges on finding
an \emph{alignment} of any $k$ traces such that
the invariant only needs to describe aligned states.
%

In the case of $k$-safety properties, an alignment of traces is often given by a \emph{self composition}~\cite{DBLP:conf/csfw/BartheDR04,DBLP:conf/sas/TerauchiA05}
of the program, composing different copies of the program (or several different programs) together, \eg, by running the different copies in lockstep~\cite{DBLP:conf/fm/ZaksP08}
or by more sophisticated composition schemes, \eg,~\cite{DBLP:conf/esop/EilersMH18}.
While self composition allows to reduce $k$-safety verification to standard safety verification,
this reduction requires to choose the alignment of the different copies a-priori.
%
%
The choice of alignment, however, has a significant effect on
the complexity of the inductive invariants themselves,
as demonstrated
by~\cite{DBLP:conf/cav/ShemerGSV19}. This
renders the standard reduction from $k$-safety verification to safety verification, based on a fixed alignment, impractical
in many cases.
As a result, finding a good alignment as part of relational verification
has been a topic of interest in recent years~\cite{DBLP:conf/pldi/SousaD16,DBLP:conf/cav/FarzanV19,DBLP:conf/cav/UnnoTK21,DBLP:conf/cav/BaumeisterCBFS21,DBLP:conf/cav/BeutnerF22}.

In the case of hyperliveness properties that stem from the use of existential
quantification over traces (\ie $\forall^*\exists^*$ properties), complexity
rises further.
Verifying such hyperliveness properties
calls for finding ``witness'' traces that match the universally quantified traces, in
addition to the relational invariant and alignment.
This reduces verification of  $\forall^*\exists^*$ properties to the problem of inferring three ingredients: 
(i)~a witness function for existential quantification over traces,
(ii)~an alignment of traces, and
(iii)~a corresponding relational inductive invariant.
These ingredients are all interdependent: different witnesses call for different alignments and give rise to different invariants,
with different levels of complexity.
It is therefore desirable to search for the \emph{combination} of the three of them \emph{simultaneously},
which is the focus of this paper.


%


We propose a novel reduction from verification of hyperproperties
with a $\forall^* \exists^*$ quantification prefix 
over infinite-state transition systems to satisfiability
of Constrained Horn Clauses (CHCs)~\cite{DBLP:conf/sas/BjornerMR13,DBLP:conf/birthday/BjornerGMR15}, also known as CHC-SAT.
Importantly, the reduction does not fix any of the aforementioned verification ingredients, in particular, the alignment, a-priori. Instead,
it is based on a CHC encoding of their joint requirements.
The unique structure of CHCs makes it possible to adopt software model checking techniques (e.g. interpolation~\cite{DBLP:conf/cav/McMillan14}, IC3/PDR~\cite{DBLP:conf/sat/HoderB12,DBLP:conf/cav/KomuravelliGC14}) for solving them. 
Our reduction, thus, allows to use
state-of-the-art CHC solvers~\cite{DBLP:conf/fmcad/FedyukovichKB17,DBLP:conf/fmcad/HojjatR18,DBLP:conf/cav/Gurfinkel22,DBLP:conf/pldi/ZhuMJ18} to achieve a highly efficient hyperproperty
verification procedure.

While it is known that safety verification can be reduced to CHC-SAT, we
are the first to show how inferring the 
combination of a witness function, a trace alignment and an inductive invariant 
for hyperproperties
of the $\forall^*\exists^*$-fragment can be reduced to CHC-SAT.

The first step of our reduction to CHC-SAT is an encoding of
the joint requirements of the 
witness-alignment-invariant ingredients 
as a set of universally quantified formulas in first-order logic (FOL) modulo theories,
where uninterpreted predicates  capture the
witness, alignment and invariant, 
and first-order theories (\eg, arithmetic and arrays) are used for modeling the transition system and the requirements.
Such an encoding has been proposed by~\cite{DBLP:conf/cav/ShemerGSV19} for the problem of
finding an invariant together with an alignment in the context of 
verification of $k$-safety properties (the universally quantified  subset of this fragment). 
We extend their FOL encoding 
to $\forall^*\exists^*$ properties, based on the game semantics introduced in~\cite{DBLP:conf/cav/BeutnerF22}.

%

%
%
%

Unfortunately, the 
resulting FOL formulas 
are beyond what modern first-order satisfiability solvers can handle due to a combination of quantifiers with theories and uninterpreted predicates.
In particular, the FOL formulas are not in the form of CHCs.
As a result, previous works~\cite{DBLP:conf/cav/ShemerGSV19,DBLP:conf/cav/UnnoTK21} that used a similar encoding could not rely on a (single) CHC-SAT query to find the alignment and invariant simultaneously. Instead,~\cite{DBLP:conf/cav/ShemerGSV19} resorted to an enumeration of potential alignments, using a separate CHC-SAT query to search for an inductive invariant (in a restricted language) for each candidate alignment.
\cite{DBLP:conf/cav/UnnoTK21} developed a specialized solver that is able to handle these non-CHC formulas directly.

In contrast to previous works, we introduce a second step where we transform the  set of universally quantified FOL formulas to a set of universally quantified CHCs. 
%
This step---
which is also the key technical contribution of the paper--- allows us to use any CHC solver 
for hyperproperty verification, and benefit from current and future developments in this lively area of research.
%
We emphasize that the transformation to CHCs is surprising since it allows us to overcome  a seemingly unavoidable 
obstacle: a disjunction of atomic formulas involving unknown predicates, which arises from the encoding of a choice between different alignment and witness options.

We implemented the reduction of $\forall^*\exists^*$-hyperproperty verification to CHC-SAT in a tool called \ours,
on top of Z3~\cite{DBLP:conf/tacas/MouraB08}, using \Spacer~\cite{DBLP:conf/cav/Gurfinkel22} as a CHC solver. Our results show that \ours is
very efficient in verifying $\forall^*\exists^*$-hy\-per\-pro\-perties, outperforming
the state-of-the-art~\cite{DBLP:conf/cav/UnnoTK21,DBLP:conf/cav/BeutnerF22,DBLP:conf/cav/ShemerGSV19} by orders of magnitude.

\smallskip\noindent
Our main contributions are:
\begin{itemize}[topsep=0pt,leftmargin=1.5em]
  \item We develop a
  satisfiability-preserving transformation of first-order formulas of a certain form to CHCs. The transformation is accompanied by a bi-directional translation of solutions.
  \item We apply the transformation to obtain, for the first time, a sound and complete reduction from verification of \aehyper (w.r.t.\ a game semantics)
  to CHC-SAT. 
    The reduction 
    captures searching for an alignment, an $\exists^*$-witness function and an inductive invariant simultaneously. It is applicable to infinite-state transition systems, with the caveat that
    their branching degree needs to be finite (bounded by a constant) if the hyperproperty includes $\exists^*$ quantification.
  \item To handle $\exists^*$ in the presence of unbounded nondeterminism, we incorporate into the CHC encoding 
  a sound abstraction 
  based on a set of underapproximations (``restrictions'').
  \item We implement a tool, \ours, that constructs CHCs for
  \aehyper specifications, and solves them using \Spacer.
  In most cases, \ours discovers the solution completely automatically, while in some, 
  it uses predicate abstraction, based on user-provided predicates.
\end{itemize}

%% file: 02-overview.tex
\section{Overview}\label{sec:overview}

We illustrate our approach for verifying hyperproperties by reduction to CHC-SAT. 
We start with the simpler 
case of $k$-safety properties, followed by the more general case of $\forall^* \exists^*$ hyperproperties.


\subsection{Motivating Example}
\label{sec:motivating}


\begin{figure}[t]
\footnotesize
\centering
\subfloat[] {
\begin{tabular}{l}
\\[1.1em]
\lstset{basicstyle=\ttfamily\scriptsize}
\begin{lstlisting}[language=C,mathescape=true]
pre($a_1<a_2 \land b_1>b_2$)
squaresSum(int a, int b){
  assume(0 < a < b);
  int c=0;
  while (a<b) {c+=a*a; a++;}
  return c;
}
post($c_1>c_2$)
\end{lstlisting}
\\[3.5em]
\fboxsep=2pt
\fbox{\scalebox{0.9}{
$
\renewcommand{\arraystretch}{1}
\begin{array}{l@{}}
a_1 < a_2 \land b_1 > b_2 \to {} \\
\quad\forall\pi_1:\lnot(a < b),\pi_2:\lnot(a < b)\cdot
\square(c_1 > c_2)
\end{array}
$
}}
\end{tabular}
\label{intro:running-example}
}
\subfloat[] {
\hspace{-1.5cm}
\begin{minipage}{0.7\textwidth}
\resizebox{1\textwidth}{!}{
$
\renewcommand\arraystretch{1.2}
\begin{array}{r@{}l}
    \Init(V_1) \land \Init(V_2) \land a_2 >a_1 \land b_2 < b_1 \rightarrow{} & Inv(V_1,V_2)  \\
    Inv(V_1,V_2)\land A_{\{1\}}(V_1,V_2) \wedge \Tr(V_1,V_1') \land V_2 = V_2'
 \rightarrow{} & Inv(V_1',V_2') \\
    Inv(V_1,V_2)\land A_{\{2\}}(V_1,V_2) \land V_1 = V_1' \land \Tr(V_2,V_2')
 \rightarrow{} & Inv(V_1',V_2') \\
    Inv(V_1,V_2)\land A_{\{1,2\}}(V_1,V_2) \land \Tr(V_1,V_1') \land \Tr(V_2,V_2')
 \rightarrow{} & Inv(V_1',V_2') \\
    Inv(V_1,V_2) \land A_{\{1\}}(V_1,V_2) \rightarrow{} & a_1 < b_1 \\
    Inv(V_1,V_2) \land A_{\{2\}}(V_1,V_2) \rightarrow{} & a_2 < b_2 \\
    Inv(V_1,V_2) \land A_{\{1,2\}}(V_1,V_2) \rightarrow{} & (a_1 < b_1 \land a_2 < b_2) \\
    & \qquad \lor (a_1 \ge b_1 \land a_2 \ge b_2)\\
    \multicolumn{2}{l}{\hspace{2cm} Inv(V_1,V_2)\rightarrow ((a_1 \geq b_1 \land a_2 \geq b_2)  \rightarrow  c_1 > c_2)}\\
    \multicolumn{2}{l}{\hspace{2cm} Inv(V_1,V_2) \rightarrow A_{\{1\}}(V_1,V_2) \lor A_{\{2\}}(V_1,V_2) \lor A_{\{1,2\}}(V_1,V_2)}
\end{array}
$
}
\end{minipage}
\label{intro:running-example:arbiter}
}

\caption{(a) A program that computes the sum of squares of integer interval $[a,b)$ with a 2-safety specification for it, and (b) its first-order encoding. } %
\vspace{-10pt}
\end{figure}

As a means for highlighting the challenges in verifying hyperproperties, and, in particular, in reducing the problem to CHC solving,
we present
the example program \code{squaresSum} and its $2$-safety specification from \cite{DBLP:conf/cav/ShemerGSV19} in \Cref{intro:running-example}.
Given positive integers $a < b$, the program computes the sum of squares of all  integers in the interval $[a,b)$.
\code{squaresSum} is monotone in the sense that as the input interval increases, so does the output $c$. Formally, this is a $2$-safety property\iflong, as it relates
two traces of the program;
it \else that \fi requires that whenever two traces satisfy the pre-condition $[a_2,b_2) \subset [a_1,b_1)$, they also satisfy
the post-condition $c_1 > c_2$, where variable indices correspond to the traces that they represent.
This is a special case of $k$-safety, where the relational property is checked at the end of the executions. More generally, \SH{rephrased:}we consider $k$-safety properties where the relational property is 
specified at designated \emph{observation points} (explained in \Cref{background:hyperproperties}).

To verify the $2$-safety property, a prominent approach is to reduce the problem to a regular safety verification problem by composing the program with itself (known as ``self composition''). There are (infinitely) many possibilities for aligning the traces in the composed system, and the alignment chosen has direct impact on the complexity of the inductive invariant needed to establish safety.
For example, if the two traces of \code{squaresSum} are aligned in lockstep, then initially  $c_1=c_2$, after one step, $c_1<c_2$, and only later on, $c_1>c_2$. Showing that $c_1>c_2$ at the end requires tracking the difference $c_1-c_2$, which is a complex value because it involves the sum of squares itself.
This cannot be captured by an inductive invariant in first-order logic, and is therefore beyond reach for state-of-the-art solvers.
On the other hand, if the second trace, whose input is the smaller interval, ``waits'' for $a_1$ and $a_2$ to coincide before proceeding in lockstep, then the property that $c_1 > c_2$ becomes inductive (except for the first step), greatly simplifying the inductive invariant. It is therefore important to consider the alignment and the (relational) inductive invariant together.

The requirements that the alignment and inductive invariant need to satisfy can be formulated in first-order logic~\cite{DBLP:conf/cav/ShemerGSV19}.
To do so,
we denote the program variables by $V = \langle a,b,c \rangle$. We  express the initial states and program steps as formulas over $V$ (and primed variant $V'$) :
\iflong
\[
\begin{array}{l@{~~}c@{~~}l}
\Init(V) & \eqdef & a >0 \land b >a \land c=0 \\
\Tr(V,V') & \eqdef & a<b \land c' = c + a\cdot a \land a' = a+1 \land b' = b
\end{array}
\]
\else
$\Init(V) \eqdef a >0 \land b >a \land c=0$,
$\Tr(V,V') \eqdef a<b \land c' = c + a\cdot a \land a' = a+1 \land b' = b$.
\fi
To reason about two traces, we use two copies of $V$, denoted $V_1$ and $V_2$.
We introduce ``unknown'' predicates $\Inv, A_{\{1\}}, A_{\{2\}}, A_{\{1,2\}}$ over $\langle V_1,V_2 \rangle$ to capture the inductive invariant and desired alignment of the traces.
$\{A_\au\}_\au$ define an \emph{arbiter} that, when $A_\au$ is satisfied, schedules the steps of the traces according to $\au$
(for example, schedule $u = \{1\}$ stands for a step in trace $1$ and a stutter in trace $2$). The arbiter therefore determines the alignment of the traces.
The inductive invariant $\Inv$  relates states of the two copies of the program, making it \emph{relational}.
%

The problem of searching for the alignment and the inductive invariant simultaneously is then posed as a satisfiability  problem (modulo the theory of arithmetic)
of the formulas in \Cref{intro:running-example:arbiter}.
To ensure that the arbiter, which determines the alignment, does not avoid violations of the post-condition by making one of the traces stutter forever s.t.\ it never reaches its final state, formulas 5-7 require that the arbiter only schedule a trace if it has not exited the loop, \SH{added:}unless both traces exited the loop (in which case both are scheduled). This ``validity'' requirement means that, at the latest, the arbiter must schedule a trace when the other reaches the final state.
Formulas 1-4 then ensure that all states that are reachable, subject to the steps permitted by the arbiter, must satisfy $\Inv$.
Specifically, the first formula ensures the initiation condition of the inductive invariant: the invariant satisfies
the pre-condition and includes all the initial states of the composed system. Formulas 2-4 ensure
the consecution of the invariant under every choice the arbiter makes.
The 8th formula ensures the safety of
the invariant and the last formula mandates that there
is always at least one choice that is enabled,
and that the system never reaches a ``stuck'' state.

%

An interpretation for the unknown predicates $\Inv, A_{\{1\}}, A_{\{2\}}, A_{\{1,2\}}$ defines an arbiter and a corresponding inductive invariant.
A possible solution is

\smallskip\noindent
$
\renewcommand*{\arraystretch}{1.2}
\begin{array}{l}
A_{\{1\}}(V_1,V_2) ~\eqdef~   a_1 < a_2 \lor  (b_2 \le a_1 < b_1)  \qquad
A_{\{2\}}(V_1,V_2) ~\eqdef~   \bot \\
A_{\{1,2\}}(V_1,V_2) ~\eqdef~   (a_1 = a_2 \land a_1<b_2)\lor a_1 \ge b_1 \\

\Inv(V_1,V_2) ~\eqdef~  0 < a_1 \leq b_1 \land 0 < a_2 \leq b_2 \land
\big((a_1 < a_2 \land c_1 \geq c_2) \lor
     (a_1 \geq a_2 \land c_1 > c_2)\big)
\end{array}
$

\smallskip
This solution captures the arbiter that makes the second trace wait until $a_1 = a_2$, \SH{elaborated here:}then makes both traces proceed together until the second one exits its loop, in which case the first trace continues to execute alone until it also exits its loop and both traces are again (vacuously) scheduled together. The solution to $\Inv$ captures the corresponding inductive invariant previously discussed.

\subsection{Challenges in Encoding Hyperproperty Verification as CHC-SAT}
The formulas of \Cref{intro:running-example:arbiter}, with the exception of the last one, are constrained \emph{Horn} clauses. That is, when the implications in these formulas are converted to disjunctions, at most one predicate application appears positively in each clause.

Alas, the presence of the last formula precludes direct application of existing CHC solvers. 
The problem is the disjunction on the right hand side of the implication.
Such a disjunction appears to be crucial for a correct encoding of the problem. 
The reason is that uninterpreted predicates designate semantic \emph{relations}.
With such predicates denoting the choice of schedule,
it is easy to drop into a vacuous solution where some states have no corresponding choice and are essentially ``stuck'', unsoundly making a post-condition violation unreachable.
Encoding the requirement that every state have a schedule results in a clause with multiple occurrences of positive literals, capturing inherent disjunctions over the possible choices, which are not Horn.
In particular, these disjunctions cannot be eliminated by renaming~\cite{DBLP:journals/jacm/Lewis78}.


Previous works tackled this obstacle either by employing explicit enumeration of alignments that satisfy the non-Horn clause to avoid the disjunction~\cite{DBLP:conf/cav/ShemerGSV19}, or by  developing specialized techniques that are able to handle such disjunctions~\cite{DBLP:conf/cav/UnnoTK21}.

\subsection{Our Approach: Transformation to CHC}
In this paper, we show that the problem of searching for an alignment together with a (relational) inductive invariant can be encoded using CHCs, allowing us to directly reduce the verification problem to CHC solving. Importantly, the reduction is both sound and complete.

A key insight of our
reduction to CHC-SAT
is the use of ``doomed'' states as a way to avoid the problematic disjunction over all choices of schedules. 
We refer to a
given state as ``doomed'' if it necessarily reaches a state that violates the
hyperproperty along \emph{every} valid alignment
(as opposed to \emph{some} in the direct encoding). 
Importantly, due to this conjunctive nature, doomed states lend themselves to a Horn encoding. If an initial state is identified as
doomed (i.e., the CHCs are unsatisfiable), then the property is violated and a counterexample can be retrieved. Otherwise, if the set of initial states
does not intersect the set of doomed states, then the hyperproperty is proved.
Moreover, given an interpretation of the unknown predicates in which the initial states are not doomed, an alignment
and a corresponding inductive invariant can be retrieved.

Based on this insight, in \Cref{transformation}, we develop
a general transformation of formulas of a certain form,  to an equi-satisfiable set of CHCs.
Furthermore, we provide a  transformation of solutions between the two formulations (in both directions).
The first-order formulas to which the transformation is applicable  follow the overall structure of the formulas in \Cref{intro:running-example:arbiter}, but are somewhat more general. For example, some of the unknown predicates may have additional arguments, which turn out to be useful 
when considering a broader class of hyperproperties beyond $k$-safety ($\forall^*\exists^*$).

In \Cref{safety} we apply the transformation of \Cref{transformation} to reduce $k$-safety verification to CHC-SAT.
When applying the transformation on the formulas encoding our running example (\Cref{intro:running-example:arbiter}), we obtain the following set of CHCs over unknown predicates $D_{\{1\}}, D_{\{2\}}, D_{\{1,2\}}$:

\begin{figure}
\resizebox{0.9\textwidth}{!}{
$
\renewcommand\arraystretch{1.2}
\newcommand\Ds[1]{D_{\{1\}\!}(#1)\land D_{\{2\}\!}(#1) \land D_{\!\{1,2\}\!}(#1)}
\begin{array}{r@{}l}
    \Ds{V_1,\!V_2} \land
    \Init(V_1) \land \Init(V_2) \land
      a_2 > a_1 \land b_2 < b_1 \rightarrow{}
        & \bot  \\
    \lnot(a_1 \geq b_1 \land a_2 \geq b_2 \rightarrow
      c_1 > c_2) \rightarrow{}
        & D_{\{1\}}(V_1,V_2) \\
    \lnot(a_1 \geq b_1 \land a_2 \geq b_2 \rightarrow
      c_1 > c_2) \rightarrow{}
        & D_{\{2\}}(V_1,V_2) \\
    \lnot(a_1 \geq b_1 \land a_2 \geq b_2 \rightarrow
      c_1 > c_2) \rightarrow{}
        & D_{\{1,2\}}(V_1,V_2) \\
    \lnot(a_1 < b_1) \rightarrow{} & D_{\{1\}}(V_1,V_2) \\
    \lnot(a_2 < b_2) \rightarrow{} & D_{\{2\}}(V_1,V_2) \\
    \lnot(a_1 < b_1 \land a_2 < b_2) \land \neg (a_1 \ge b_1 \land a_2 \ge b_2) \rightarrow{}
        & D_{\{1,2\}}(V_1,V_2) \\
    \Ds{V_1',V_2'} \land
      \Tr(V_1,V_1') \land V_2 = V_2' \rightarrow{}
        & D_{\{1\}}(V_1,V_2) \\
    \Ds{V_1',V_2'} \land
      V_1 = V_1' \land \Tr(V_2,V_2') \rightarrow{}
        & D_{\{2\}}(V_1,V_2) \\
    \Ds{V_1',V_2'} \land
      \Tr(V_1,V_1') \land \Tr(V_2,V_2') \rightarrow{}
        & D_{\{1,2\}}(V_1,V_2) \\
\end{array}
$
}
\caption{CHC encoding of \Cref{intro:running-example}.}
\label{safety:running-example:doomed}
\end{figure}


Here, an unknown predicate $D_u$ represents states that are ``doomed'' if schedule $u$ is chosen. The first CHC requires that no initial state that satisfies the pre-condition is  completely doomed, i.e., for every such state there is a schedule for which it is not doomed.
The remaining CHCs encode the properties of doomed states for each schedule.
For example, the CHCs where $\D_{\{1\}}$ is in the head (right hand side of the implication) imply that a state is doomed for schedule $\{1\}$ if: (a)~it violates the post-condition,
(b)~it already exited the loop and hence trace $1$ cannot be the only trace to be scheduled, or (c)~it is the pre-state of a transition taken by $1$ leading to a post-state that is doomed for every choice $u$.


A solution to the CHCs in \Cref{safety:running-example:doomed} can be obtained from the solution to the formulas in \Cref{intro:running-example:arbiter} by $D_{u} \eqdef \lnot(\Inv\land A_{u})$ for every $u \in \{ \{1\}, \{2\}, \{1,2\}\}$.
\SH{suggestion: use space in doomed fig to include solution for $A, Inv$ and underneath solution for $D$}

More generally, in \Cref{transformation}, we show a bi-directional transformation of solutions.

\subsection{Beyond $k$-Safety}

Our transformation to CHCs is not limited to an encoding of $k$-safety, but also generalizes to  hyperproperties that use $\forall^* \exists^*$ quantification over traces, as presented in \Cref{sec:beyond}.

Hyperproperties with existential trace quantification become meaningful in the presence of nondeterminism in the program. 
For an example of such a property, consider a nondeterministic variant of \code{squaresSum} where the assignment \code{c += a * a} is replaced by \code{if (*) c += a * a}. That is, the increment of \code{c} may nondeterministically be skipped.
We may now wish to verify that, if $[a_2,b_2) = [a_1,b_1)$, then for every trace from input $[a_1,b_1)$ there exists a trace from input $[a_2,b_2)$ such that when both terminate, $c_1 \neq c_2$. This is a $\forall \exists$-hyperproperty.

To verify such properties, a ``witness'' function is needed to map the universally quantified traces to the corresponding existentially quantified traces such that the body of the formula holds for the combination of the traces.
Even if a witness function is known, to verify that the combination of the traces satisfies the body of the formula, we  still need to find a proper alignment of the traces and an inductive invariant. As in the case of $k$-safety, these components are all interdependent, making it desirable to search for all of them together.

In general, the witness function for the existentially quantified traces may need to depend on the full universally quantified traces. However,~\cite{DBLP:conf/cav/BeutnerF22} define a sound but incomplete \emph{game semantics}, in which the witness function essentially constructs the existentially quantified  traces step-by-step, in response to moves of a ``falsifier'' who reveals the universally quantified traces step-by-step.

We show in \Cref{sec:beyond:finite} that the problem of searching for a step-by-step witness function, an alignment and a (relational) inductive invariant can be encoded in first-order logic, and the encoding is amenable to our transformation to CHCs. This results in a sound and complete CHC encoding of the game semantics of~\cite{DBLP:conf/cav/BeutnerF22} for transition systems with finite branching.

The idea in the $\forall^* \exists^*$-first-order encoding is to let the unknown predicates $A_u$  specify not only the schedules chosen by the arbiter but also the choice of existentially quantified traces for the witness function.
To do so, we assign a unique label to each of the possible transitions, and use these labels to identify the transitions along the traces.
In this encoding, instead of $u$ denoting a schedule only, it now denotes both a schedule and a choice of labels identifying the next transitions in the existentially quantified traces according to the witness function.
%
Furthermore, the $A_u$ predicates receive additional arguments that represent the next labels along the universally quantified traces.

For example, in the nondeterministic variant of \code{squaresSum}, there are at most two possible transitions in each control location. We therefore introduce two labels to distinguish between these possibilities: \code{i} for ``increment'' and \code{s} for ``skip''.
The predicates that describe the schedules and the choices of existentially quantified traces for the $\forall \exists$-hyperproperty of interest are
$A_{\{1\},\text{\code{i}}}, A_{\{2\},\text{\code{i}}}, A_{\{1,2\},\text{\code{i}}}, A_{\{1\},\text{\code{s}}}, A_{\{2\},\text{\code{s}}}, A_{\{1,2\},\text{\code{s}}}$. They are defined over $\langle V_1,V_2, \Lblone \rangle$, where $\Lblone$ ranges over the possible labels.

Note that in this encoding, the $A_u$ predicates are no longer defined over $\langle V_1,V_2 \rangle$ only, but have additional arguments for the labels of the universally quantified traces, while $\Inv$ does not. Thus, the reduction to CHCs applies our transformation in a more general setting than \Cref{intro:running-example:arbiter}.
Furthermore, since $u$ denotes both a schedule and a choice of labels for the existentially quantified traces,
the number of $A_u$ predicates depends on the number of labels. To ensure that there are finitely many predicates, we require the transition system to have a finite branching degree (otherwise, the space of possible labels becomes infinite).

Finally, in \Cref{sec:beyond:infinite}, we extend our approach to handle infinite branching in the transition system, which can result, for example, from reading an input from an infinite domain.
To do so, we introduce another first-order encoding that roughly replaces the infinitely-many concrete choices of transitions by finitely-many abstract choices. Unlike the cases of $k$-safety and $\forall^* \exists^*$-hyperproperties with finite branching, the resulting  encoding is sound but incomplete w.r.t. the game semantics.
By applying our transformation, we obtain a sound (albeit incomplete) reduction to CHC solving.




%% file: 03-background.tex
\section{Background}
\label{background}
In this section we provide the necessary background on first order logic modulo theories, hyperproperties, and constrained Horn clauses.

\paragraph{First Order Logic}
We use many-sorted first-order logic to model systems and their properties. We assume the reader is familiar with the syntax and semantics of first-order logic. \SH{added:}A first-order theory $\theory$ is a set of formulas (usually closed under first-order deduction). Models of $\theory$ are first-order structures that satisfy all of the formulas in $\theory$.
Throughout the paper, we fix a background first-order theory $\theory$ and denote its signature by $\Sigma$.
Sorts and symbols in $\Sigma$ are called \emph{interpreted}.
For example, $\theory$ can be linear integer arithmetic (LIA), in which case $\Sigma$ includes sort $\sort{int}$, interpreted constant symbols $\ldots,-1,0,1,\ldots$ for the integers, an interpreted function symbol $+$ for addition, an interpreted predicate 
symbol $<$, etc.
The background theory $\theory$ can also be the combination of theories (\eg, LIA combined with the theory of arrays).
Whenever we consider satisfiability of a formula (or a set of formulas), we mean satisfiability modulo $\theory$; a formula $\varphi$ is satisfiable modulo $\theory$ if there exists a model of $\theory$ that satisfies $\varphi$. Similarly, by validity we refer to validity modulo $\theory$: $\varphi$ is valid modulo $\theory$ if every model of $\theory$ satisfies $\varphi$.


\paragraph{Transition Systems}
A (semantic, labeled) \emph{transition system} is a tuple $\TSsem = (\States, \Lbl,\InitSem, \TrSem)$, where $\States$ is a  set of states; $\Lbl$ is a set of labels; $\InitSem \subseteq \States$ is the set of initial states; and
$\TrSem \subseteq \States \times \Lbl \times \States$ is the transition relation.
A trace of $\TSsem$ is a maximal 
sequence of states $t = s_0,s_1,\ldots$ such that 
for every $i \geq 0$ there exists $\ell \in \Lbl$ such that $(s_i,\ell,s_{i+1})\in \TrSem$.
We denote by $t[i]$ the $i$'th state in $t$.
We further denote the set of traces that start from  a state $s$ by $\Traces(s)$,
and the set of all traces of $\TSsem$ by $\Traces$.

We usually consider transition systems that are defined symbolically.
A (symbolic, labeled) transition system is a tuple $\TS = (\Vone, \Lblone, \Init, \Tr)$, where $\Vone$ is a vocabulary, \ie, a vector of (logical) variables, each associated with a sort from $\Sigma$, denoting state variables; $\Lblone$ is a label variable, with a designated sort; $\Init$ is a formula over $\Sigma$ with free variables $\Vone$, and $\Tr$ is a formula over $\Sigma$ with free variables $\Vone\cup\{\Lblone\}\cup \Vone'$,
where $\Vone'$ consists of the primed variants of $\Vone$. 
%
The semantics of $\TS$ is given by a (semantic) transition system $\TSsem = (\States, \Lbl,\InitSem, \TrSem)$ with the following correspondence:
each state in $\States$ is a valuation to $\Vone$; the set of labels $\Lbl$ consists of the set of values that $\Lblone$ can take; the set of initial states $\InitSem \subseteq \States$ consists of all  valuations to $\Vone$ that satisfy $\Init$; and
the transition relation $\TrSem \subseteq \States \times \Lbl \times \States$ consists of the valuations for the composite vocabulary $\Vone\cup\{\Lblone\}\cup \Vone'$ that satisfy $\Tr$. 
For simplicity, we consider transition systems $\TS$ in which 
$\TrSem$ is total, \ie, $\forall s \in \States \ \exists \ell \in \Lbl, s' \in \States \cdot (s,\ell,s')\in \TrSem$\footnote{w.l.g.; $\Tr$ can always be replaced by  $\Tr \vee ((\forall \Lblone \ \forall \Vone' \cdot \neg \Tr) \wedge \Vone' = \Vone)$, which corresponds to adding self loops to states that have no outgoing transition.} (this assumption ensures that all traces are infinite).
We say that $\TS$ is \emph{deterministic} when $\forall s\in\States,\ell\in\Lbl\cdot\big|\{s'\,|\,\TrSem(s,\ell,s')\}\big|= 1$.

In the sequel, unless explicitly stated otherwise, we consider symbolic transition systems. 
By convention, $\TSsem$ will always represent the semantic counterpart of $\TS$.

\subsection{Hyperproperties and Their Specification}
\label{background:hyperproperties}
To express hyperproperties, we consider a fragment of the relational logic OHyperLTL~\cite{DBLP:conf/cav/BaumeisterCBFS21}, 
which we call \aehyper. 
Formulas in this fragment include $\forall^* \exists^*$ quantification over traces, and specify a relational property $\phi$ that needs to hold \emph{globally}
along the examined traces when the traces reach certain \emph{observation points} (hence the ``O'' in OHyperLTL). Observation points for each of the traces are determined by (non-temporal) formulas $\xi_i$. As shown in~\cite{DBLP:conf/cav/BaumeisterCBFS21}, observation points increase the expressiveness of the logic, as they allow to consider asynchronous executions of the different traces (whereas in HyperLTL~\cite{DBLP:conf/post/ClarksonFKMRS14}, all traces execute synchronously).

\paragraph{Syntax} Formally, given a vector of variables $\Vone$, formulas in 
\aehyper are of the form:
%
\[
\varphi ~=~ \psi \rightarrow \forall \pi_1 : \xi_1, \ldots, \pi_l : \xi_l\cdot \exists \pi_{l+1} : \xi_{l+1},\ldots,\pi_k:\xi_k\cdot~ \square\phi
\]
%
\noindent
where $\pi_i$ are trace variables whose intended valuations are taken from $\Traces$, $\xi_i$ are (non-temporal) formulas with free variables $V$ that determine \emph{observation points} along each of the $k$ traces,
and $\psi,\phi$ are (non-temporal) formulas with free variables $\Vone_1 \cup\ldots \cup \Vone_k$, expressing relational properties for the $k$ traces.
$\psi$ is a pre-condition that is assumed to hold initially, and
$\phi$ needs to globally hold at the observation points of the traces.
$V_j$ denotes a copy of $V$ where all variables are indexed by $j$.
We refer to the variables in $V_j$ as the state variables of the $j$'th trace (namely, $\pi_j$).
%

When $l = k$, \ie, all quantifiers are universal, $\varphi$ specifies a \emph{$k$-safety} property. A relational pre/post specification, as used in our motivating example, is a special case of a $k$-safety property where the observable points are the final states (which are augmented with self loops). For example, \Cref{intro:running-example} presents the \aehyper specification of the motivating example.

When $l<k$, the formula also includes existential quantifiers, extending expressiveness to include some hyperliveness properties.
An example of a security hyperliveness property that can be expressed in \aehyper
is \emph{generalized non-interference (GNI)}~\cite{DBLP:conf/sp/McCullough88}.
GNI requires that for any two traces $\pi_1$ and $\pi_2$ there exists a trace $\pi_3$ whose high (secret) inputs agree with $\pi_1$ and whose low (public) inputs and outputs agree with $\pi_2$.
GNI can be expressed in \aehyper using the following formula:
\[
\true \rightarrow \forall \pi_1 : \texttt{obs} \ \forall \pi_2 : \texttt{obs} \ \exists \pi_3: \texttt{obs} \cdot~ \square (h_1 = h_3 \land l_2 = l_3 \land o_2 = o_3)
\]
where $\texttt{obs}$ denotes the observable events along each trace.

\paragraph{Semantics}
\aehyper formulas are interpreted over transition systems.
Intuitively, $\varphi = \psi \rightarrow \forall \pi_1 : \xi_1, \ldots, \pi_l : \xi_l\cdot \exists \pi_{l+1} : \xi_{l+1},\ldots,\pi_k:\xi_k\cdot~ \square\phi$ holds in a transition system
if from every $k$ initial states that jointly satisfy the pre-condition $\psi$,
for every $l$ traces from the first $l$ states there exist corresponding $k-l$ traces from the remaining $k-l$ states
s.t.\ the composed states of all traces globally satisfy $\phi$, when the traces are projected to their observation points.
%

Formally, given a transition system $\TS$ and $\varphi$ as above,
we refer to a tuple $(s_1,\ldots,s_k)$ of $k$ states of $\TS$  as a \emph{composed state}.
A composed state defines a valuation to $V_1\cup \ldots \cup V_k$, where $s_j$ is the valuation of $V_j$.
A composed state is initial if $s_i \in \InitSem$ for every $1 \leq i \leq k$.

We say that $\TS \models \varphi$ if
for every initial composed state $\overline{s} = (s_1,\ldots,s_k)$ such that $\overline{s} \models \psi$ the following holds:
for every  $t_1,\ldots,t_l \in \Traces(s_1) \times \cdots \times \Traces(s_l)$  there exist   $t_{l+1},\ldots,t_{k} \in \Traces(s_{l+1})\times \cdots \times \Traces(s_k)$ such that
$\observe{t_1}{\xi_1},\ldots,\observe{t_k}{\xi_k}\models\square\phi$, where $\observe{t_i}{\xi_i}$ is the projection (filtering) of trace $t_i$ to states satisfying $\xi_i$.
The semantics of $\square\phi$ is that $t_1',\ldots,t_k'\models\square\phi$ iff
$\forall i\leq \text{min}(|t_1'|,\ldots,|t_k'|)
\cdot (t_1'[i],\ldots,t_k'[i])\models\phi$, where $t_j'[i]$ denotes the $i$'th state in the sequence $t_j'$.
%

Note that the semantics is oblivious to the transition labels
since labels are only implicit in traces.
Labels will become useful in \Cref{sec:beyond}, where we will use them to identify transitions along traces.

\iflong
\begin{remark}
We define the semantics of hyperproperties w.r.t.\ a single transition system
for simplicity of the presentation only.
It is straightforward to extend the definition, as well as the approach presented in this paper,
to the case where each trace quantifier refers to traces of a different transition system.
\end{remark}

\begin{remark}
Our approach generalizes to the case where $\square \phi$ is replaced by any temporal safety property
via the standard automata-theoretic approach to model checking,
where a finite-state monitor for the property is composed with the transition system and the property reduces to checking that the error state of the monitor is never reached. 
\SH{remove for space?}(Note that even though a safety property is used for the body of the formula, the property as a whole is not necessarily a safety property due to the existential quantification over traces~\cite{DBLP:conf/cav/BeutnerF22}.)
\end{remark}

\else
\begin{remark}
To simplify the presentation we consider hyperproperties defined w.r.t.\ a single transition system.
The extension to multiple transition systems is straightforward.
Similarly, $\square \phi$ can be generalized to any temporal safety property
via the standard automata-theoretic approach to model checking.
\end{remark}
\fi

\subsection{Constrained Horn Clauses (CHCs)}
Given a background theory $\theory$ over signature $\Sigma$, constrained Horn clauses are
defined over a signature $\Sigma'$ that extends $\Sigma$ with a set $\CHCPreds$ of (uninterpreted) predicates.
As opposed to the symbols in $\Sigma$, which are interpreted by $\theory$, the predicates in $\CHCPreds$ are \emph{uninterpreted} (sometimes called \emph{unknown}).
In the context of CHCs, first-order formulas over $\Sigma$ are often called \emph{constraints}.
A CHC is a first-order formula of the form
\[
\forall \X \cdot \bigwedge_i P_i(\X_i) \wedge \varphi(\X) \to H(\X_H)
\]
where $\X$ is a vector of (logical) variables;
$P_i \in \CHCPreds$ (not necessarily distinct, \ie, it is possible that $P_{i_1} = P_{i_2}$ for $i_1 \neq i_2$);
$H$ is either $\bot$ or a predicate from $\CHCPreds$;
$\X_i, \X_H \subseteq \X$; and $\varphi$ is a constraint\iflong (\ie, a first-order formula over $\Sigma$)
whose free variables are a subset of $\X$\fi.
%
The universal quantification over $\X$ is often omitted.
A system of CHCs is a finite set of CHCs.

\emph{Satisfiability} of a set of CHCs (modulo $\theory$) is defined in the usual way (as first-order formulas).
CHC solvers attempt to determine the satisfiability of a set of CHCs by searching for
a satisfying model of the CHCs where the interpretations of the unknown predicates, $\CHCPreds$, are definable in $\theory$.
Such a model is called a solution.
Formally, a \emph{solution} to a set of CHCs maps every predicate in $\CHCPreds$
to a formula over $\Sigma$ that defines it
such that substituting all occurrences of the predicates by their definitions
results in formulas that are valid modulo $\theory$.
If a set of CHCs has a solution then it is satisfiable. However,
the converse may not hold due to limited expressivity of first-order logic.

%% file: 04-transformation.tex
\section{General Transformation to CHCs}
\label{transformation}

\textit{This section describes the technical details of a satisfiability-preserving transformation to CHCs; it can be safely skipped without hindering the reader's understanding of the following material.
Only \Cref{thm:trans} is used later to uphold our soundness guarantee.}

\smallskip
The transformation described herein lets us translate
a set of formulas, which adheres to a specific FOL scheme,
to an equi-satisfiable set of CHCs.
Later we show how verification of a \aehyper property can be captured by a set of formulas of the aforementioned scheme,
and use the described transformation to retrieve a set of CHCs.
This allows us to then reason about the correctness of the \aehyper property by deciding the satisfiability 
of the CHCs.



\begin{figure}[t]
\hspace{-0.05\textwidth}
\resizebox{1.1\textwidth}{!}{
$
\renewcommand{\arraystretch}{1.2}
\begin{array}{l|l}
\subfloat[] {$
\begin{array}{l@{\,}r@{}l}
    & \alpha(\V) \rightarrow{} & \Inv(\V)  \\
    & \Inv(\V) \land \beta(\V) \rightarrow{} & \bot \\
    \eqForall & \Inv(\V) \wedge A_\au(\V,\W) \land \gamma_\au(\V,\W) \rightarrow{} & \bot \\
    \eqForall & \Inv(\V)\land A_\au(\V,\W) \land \delta_\au(\V, \V', \W)
 \rightarrow{} & \Inv(\V') \\

    & \Inv(\V) \rightarrow{} & \displaystyle \bigvee_{\au\in\U} A_\au(\V,\W)
\end{array}
$
\label{transformation:formulas:arbiter}
}
&
\subfloat[] {$
\begin{array}{l@{}r@{}l}
    & \displaystyle \bigwedge_{\au\in\U} \!\!\D_\au(\V,\W)
    \land \alpha(\V) \rightarrow{} & \bot \\[-0.5em]
    \eqForall & \beta(\V) \rightarrow{} & \D_\au(\V,\W) \\
    \eqForall & \gamma_\au(\V,\W) \rightarrow{} & \D_\au(\V,\W) \\
    \eqForall & \displaystyle \bigwedge_{\au'\in\U}\!\! \D_{\au'}(\V',\W')  \land \delta_\au(\V, \V', \W)
 \rightarrow{} & \D_\au(\V,\W)
 \end{array}
$
\label{transformation:formulas:doomed}
}
\end{array}
$}
\vspace{-0.5cm}

\raggedright\scalebox{0.8}{$\big(\hspace{1pt}\eqForall = \forall u\in\U\big)$}
\caption{Formula scheme before (a) and after (b) the transformation.}
\label{transformation:formulas}
\end{figure}

Consider the scheme in \Cref{transformation:formulas:arbiter} for a set of formulas over
a signature $\Sigma'$ that extends the signature $\Sigma$ of the background theory by unknown predicates $\Inv$ and $\{A_{\au}\}_{\au\in\U}$, for some finite set $U$.
$\V,\V',\W$ denote disjoint vocabularies, \ie, vectors of (logical) variables that are implicitly universally quantified.
%

A row prefixed by $\eqForall$ indicates $|\U|$ formulas, where $u$ is substituted by all corresponding values from $\U$.
\iflong The symbols \fi $\alpha, \beta, \gamma_\au, \delta_\au$ designate \emph{constraints}\iflong: formulas over the appropriate vocabularies using interpreted symbols only, with no occurrences of $\Inv$ or $A_\au$.\else (no occurrence of $\Inv$ or $A_\au$). \fi
%
%

At a high level, formulas 1 and 4 in \Cref{transformation:formulas:arbiter} use $\Inv$ to capture an inductive invariant of the ``states'' (valuations to $\V$)  reachable from $\alpha$ by ``transitions'' of $\delta_\au$, restricted according to a choice $\au \in \U$ of an ``arbiter'' $\{A_\au\}_\au$. Formula 2 establishes the fact that the reachable states are disjoint from some ``bad states'' $\beta$. Formulas 3 allow to enforce that the arbiter meets certain requirements, and formula 5 ensures that the arbiter makes a choice for every ``state'' in $\Inv$.

\begin{example}
For our running example, we have $\V=\langle V_1,V_2\rangle = \langle a_1,b_1,c_1,a_2,b_2,c_2\rangle$, $\V'=\langle V_1',V_2'\rangle = \langle a_1',b_1',c_1',a_2',b_2',c_2'\rangle$, and $\W=\langle\rangle$
(The extra vocabulary $\W$ will come into use later in the paper).
$\U$ is the set of arbitration choices
$\{\{1\},\{2\},\{1,2\}\}$,
and the corresponding completion of the constraint holes $\alpha, \beta, \gamma_\au, \delta_\au$ is easily discernible.\SI{or is it?}
(Note that a constraint on the right of $\to$
corresponds to its negation on the left.)
\end{example}


Note that the last formula in \Cref{transformation:formulas:arbiter} is not a CHC since
its head is a disjunction of unknown predicates. Our goal is to transform the set of formulas in \Cref{transformation:formulas:arbiter} into an equi-satisfiable set of CHCs.
%
To do so, we perform a stepwise transformation of the formulas in \Cref{transformation:formulas:arbiter} that results in the system of CHCs in \Cref{transformation:formulas:doomed}.

The first intuition is that, since every model satisfies
$\Inv(\V) \rightarrow \bigvee_{\au\in\U} A_\au(\V,\W)$
(the last formula), we can replace $\Inv(\V)$ in \Cref{transformation:formulas:arbiter} with
$\Inv(\V)\land \bigvee_{\au\in\U} A_\au(\V,\W)$,
or, equivalently,
$\bigvee_{\au\in\U} \Inv(\V)\land A_\au(\V,\W)$.
We apply this transformation to the first two formulas and to the right-hand side of the fourth one to obtain
the formulas in \Cref{transformation:formulas-arbiter_step1}.

\begin{figure}
\begin{tabular}{@{}c|@{}c@{}}
\begin{minipage}{0.57\textwidth} 
\scalebox{0.75}{
$
\renewcommand{\arraystretch}{1.2}
\begin{array}{l@{\,}r@{}l}
    & \alpha(\V) \rightarrow{} & \displaystyle \bigvee_{\au\in\U} \Inv(\V)\land A_\au(\V,\W)  \\
    & \displaystyle\left(\bigvee_{\au\in\U} \Inv(\V)\land A_\au(\V,\W)\right) \land \beta(\V) \rightarrow{} & \bot \\
    \eqForall & \Inv(\V) \wedge A_\au(\V,\W) \land \gamma_\au(\V,\W) \rightarrow{} & \bot \\
    \eqForall & \Inv(\V)\land A_\au(\V,\W) \land \delta_\au(\V, \V', \W)
 \rightarrow{} & \displaystyle \bigvee_{\au\in\U} \Inv(\V')\land A_\au(\V',\W') \\

    & \Inv(\V) \rightarrow{} & \displaystyle \bigvee_{\au\in\U} A_\au(\V,\W)
\end{array}
$
}

\vspace{-1.5em}
\raggedright\scalebox{0.8}{$\big(\hspace{1pt}\eqForall = \forall u\in\U\big)$}
\caption{First step of CHC transformation. \label{transformation:formulas-arbiter_step1}}
\end{minipage}
&
\begin{minipage}{0.45\textwidth}
\scalebox{0.75}{
$
\renewcommand{\arraystretch}{1.2}
\begin{array}{l@{\,}r@{}l}
    \\
    & \alpha(\V) \rightarrow{} & \displaystyle \bigvee_{\au\in\U} \AInv_\au(\V,\W)  \\
    \eqForall & \AInv_\au(\V,\W) \land \beta(\V) \rightarrow{} & \bot \\
    \eqForall & \AInv_\au(\V,\W) \land \gamma_\au(\V,\W) \rightarrow{} & \bot \\
    \eqForall & \AInv_\au(\V,\W) \land \delta_\au(\V, \V', \W)
 \rightarrow{} & \displaystyle \bigvee_{\au\in\U} \AInv_\au(\V',\W') \\[3em]
\end{array}
$
}
\caption{Second step of CHC transformation. \label{transformation:formulas-ainv}}
\end{minipage}
\end{tabular}
    
\end{figure}

The second formula now has a disjunction on the left-hand side of the implication.
While this is not in Horn form, it can be equivalently transformed to $|U|$ CHCs in the following manner:
$$ \forall \au\in\U \cdot{}  \Inv(\V)\land A_\au(\V,\W) \land \beta(\V) \rightarrow{}  \bot $$

It now becomes evident that $\Inv$ only ever occurs conjoined with some $A_\au$, except for the last formula.
However, a closer look reveals that the last formula is now redundant w.r.t.\@ satisfiability of the entire system:
since $\Inv$ only occurs in conjunction with $\bigvee_{\au\in\U} A_\au(\V,\W)$ on the right-hand side of implications, any model that satisfies the other formulas (all besides the last one) can be modified into one that satisfies the last formula as well simply by conjoining the interpretation of $\Inv$ with that of $\forall \W \cdot \bigvee_{\au\in\U} A_\au(\V,\W)$\yv{Need to be specific about this universal quantifier}.
We can therefore create an equisatisfiable set of formulas by redefining $\Inv(\V)\land A_\au(\V,\W)$ as new uninterpreted predicates $\AInv_\au(\V,\W)$, and dropping the last formula from \Cref{transformation:formulas-arbiter_step1}.
This results in the formulas in \Cref{transformation:formulas-ainv}.

We prove that 
\Cref{transformation:formulas:arbiter} and 
\Cref{transformation:formulas-ainv} are equisatisfiable
by 
model transformations: in one direction, define 
$\AInv_\au(\V,\W) \eqdef \Inv(\V)\land A_\au(\V,\W)$, and in the other 
$\Inv(\V) \eqdef \forall \W\cdot\bigvee_{\au\in\U}\AInv_\au(\V,\W)$ and 
$A_\au(\V,\W) \eqdef \AInv_\au(\V,\W)$.

The set of formulas 
in \Cref{transformation:formulas-ainv} is still not Horn. However, a closer look reveals that it is ``co-Horn''. Namely, instead of having at most one positive literal, each formula has at most one negative literal.
The last step is to transform the co-Horn set into a proper set of CHCs. This is done by
negating each implication and apply the following renaming~\cite{DBLP:journals/jacm/Lewis78}
$\D_\au \eqdef \lnot\AInv_\au$.
This transformation results in the 
set of CHCs 
in \Cref{transformation:formulas:doomed}.



This last transformation step is also satisfiability-preserving. Moreover, it admits a bi-directional model transformation
by negating the solutions for $\AInv_\au$, respectively $\D_\au$.
More precisely, models of \Cref{transformation:formulas-ainv} can be transformed into models of \Cref{transformation:formulas:doomed} by defining $\D_\au(\V,\W) \eqdef \neg \AInv_\au(\V,W)$
and in the other direction by defining $\AInv_\au(\V,W) \eqdef \neg \D_\au(\V,\W)$. Altogether, we conclude that the transformation from
\Cref{transformation:formulas:arbiter} to \Cref{transformation:formulas:doomed} preserves (un)satisfiability, and also allows for bidirectional translation of solutions between the original formulas and the transformed system of CHCs:


\begin{theorem}\label{thm:trans}
The set of formulas in \Cref{transformation:formulas:arbiter} is equi-satisfiable to the system of CHCs in \Cref{transformation:formulas:doomed}.
Furthermore, there is an efficient translation of models of the former to models of the latter,
and vice versa.
\end{theorem}

\begin{proof}
We have already shown the CHCs in \Cref{transformation:formulas:doomed} are obtained from the formulas in \Cref{transformation:formulas:arbiter} by a stepwise transformation, where each step preserves equi-satisfiability and models.
We obtain the final translations between models by composing the aforementioned transformations of models of \Cref{transformation:formulas:arbiter} to models of \Cref{transformation:formulas-ainv} and vice versa, and the transformations of models of \Cref{transformation:formulas-ainv} to models of \Cref{transformation:formulas:doomed} and vice versa.
The final translations, which we have verified with Z3, are:

\smallskip
$
\renewcommand{\arraystretch}{1.2}
\begin{array}{l@{\quad}|@{\quad}l}
\mbox{Given~}\Inv,A_\au\models\mbox{[Fig. \ref{transformation:formulas:arbiter}]}
&
\mbox{Given~}D_\au\models\mbox{[Fig. \ref{transformation:formulas:doomed}]} \\ \hline \\[-11pt]
D_\au(\V,\W) \eqdef \neg (\Inv(\V)\land A_\au(\V,\W))
&
\Inv(\V) \eqdef \forall \W\cdot\bigvee_{\au\in\U} \neg D_\au(\V,\W) \\
&
A_\au(\V,\W) \eqdef \neg D_\au(\V,\W)
\end{array}
$
\end{proof}


\SH{try to improve refs to CHCs numbers}%
After the transformation, $\D_\au$ in \Cref{transformation:formulas:doomed} can be understood as capturing states that are \emph{doomed} if $\au$ is chosen, in the sense that from these states the arbiter will not be able to avoid reaching bad states
(hence the choice of notation).
This is captured by CHCs 2 and 4. CHCs 3 make sure that a choice that is not allowed is also doomed. CHC 1 requires that the states that are doomed for every choice $\au$ are disjoint from $\alpha$, which ensures that, starting from $\alpha$, the arbiter always has a way to stay in the ``safe zone'' and avoid bad states.

\begin{example}
Using this schema, and from
the formulas in \Cref{intro:running-example:arbiter},
we obtain the CHCs that were already presented in~%
\Cref{safety:running-example:doomed}.

\end{example}

%% file: 05-safety.tex
\section{Encoding $k$-Safety Verification as CHCs}
\label{safety}

In this section we address the problem of verifying $k$-safety properties via a CHC encoding. To this end, we start with a natural, non-Horn, encoding of the problem, 
which is a slight generalization of the one described in \Cref{sec:motivating}, based on \cite{DBLP:conf/cav/ShemerGSV19},
and apply our transformation to obtain an equi-satisfiable system of CHCs.
\SH{added:}We note that the $k$-safety case, being more specific, can also be solved via the general $\forall^*\exists^*$ technique presented in \Cref{sec:beyond} (see also \Cref{rem:safety-as-aehyper}); but to make the presentation easier to follow, we describe the simpler case first and then generalize.

Throughout the section we fix a $k$-safety formula: 
\[
\varphi ~=~ \psi \rightarrow \forall \pi_1 : \xi_1, \ldots, \pi_k : \xi_k \cdot~ \square\phi
\]
%
This formula holds in a transition system $\TS$ if, starting from initial composed states that satisfy the pre-condition $\psi$, the observable states along every tuple of $k$ traces
satisfy $\phi$, when the observable states are reached synchronously.
\SH{remove since moved to background: Note that a pre-post specification, as used in our motivating example, is a special case of such a formula where the observable states are the final states.}%
Verifying \iflong such a property
\else $\varphi$
\fi corresponds to finding (1)~an alignment (``self composition'') of the traces that synchronizes the observation points defined by $\xi_1,\ldots,\xi_k$,
and (2)~an inductive invariant that establishes that $\phi$ holds whenever $\xi_1,\ldots,\xi_k$ hold.
Note that, while the purpose of the inductive invariant is to ensure that $\phi$ holds at the observation points, the way to guarantee this property is to require that the invariant is
inductive in all states along the aligned traces, including intermediate states between observable points (this is similar to the way  inductive invariants are used in proofs of partial correctness)\yv{This is a subtle point. We need to discuss it a little bit more and maybe elaborate. In essence, we may want to put more emphasis on the fact that inductiveness is not just between synchornization points}\SH{elaborated more - is it better?}.
\iflong
Different alignments may give rise to different inductive invariants. Thus, it is desirable to find an alignment and an inductive invariant simultaneously~\cite{DBLP:conf/cav/ShemerGSV19}.
\else
As different alignments  give rise to different inductive invariants,
it is desirable to find both of them
simultaneously~\cite{DBLP:conf/cav/ShemerGSV19}.
\fi

\SH{added}In the rest of the section, we first formalize the verification of $k$-safety properties via self composition based on the definitions of~\cite{DBLP:conf/cav/ShemerGSV19}.
We then encode the problem in first-order logic, and apply our transformation to obtain a CHC encoding.

\subsection{$k$-Safety Verification via Self Composition}

In this section we present the reduction of $k$-safety verification to safety verification based on self-composition.
To this end, we fix a transition system $\TS$ and its semantic counterpart $\TSsem= (\States, \Lbl,\InitSem, \TrSem)$ and use it throughout this subsection.
Following~\cite{DBLP:conf/cav/ShemerGSV19}, we model the alignment (or self composition) of $k$ traces of $\TS$ using an arbiter $\arbiter$ that, at every step, schedules a subset $\varnothing \neq M \subseteq \{1,\ldots,k\}$ of the traces to make a step based on the composed state  $\overline{s}=(s_1,\ldots, s_k) \in \States^k$, which consists of the current state of each trace.
Formally:
\begin{definition}[Arbiter]\label{def:arbiter}
Let $\schedules =\mathcal{P}(\{1,\ldots,k\})\setminus \{\varnothing\}$ denote the set of possible schedules for $k$ traces.
An \emph{arbiter}\footnote{Technically, the notion of an arbiter is parameterized by $k$---the number of traces considered. In our case, $k$ is fixed (since $\varphi$ is fixed), hence we simplify the definition.} for $\TS$ is a function $\arbiter : \States^k \to \mathcal{P}(\schedules)$.
\end{definition}
That is, $\schedules$ is the set of all schedules $\varnothing \neq M \subseteq \{1,\ldots,k\}$ and an arbiter
maps each composed state $\overline{s}$ to a set of possible schedules $\arbiter(\overline{s}) \subseteq \schedules$.
(Later on, when we discuss valid arbiters, we impose additional restrictions on the possible schedules.)
%

An arbiter $\arbiter$ induces a composed transition system $\TSsem^k_{\arbiter}$ over the set of composed states, where the initial states are initial composed states (i.e., $k$-tuples of initial states of $\TS$),
and the outgoing transitions of a composed state $\overline{s}$ correspond to a parallel execution of the outgoing transitions of the individual states according to the schedule $M$ determined by $\arbiter$ for $\overline{s}$.
%
Formally, $\TSsem^k_{\arbiter}$ is defined as follows:

\begin{definition}[Composed Transition System]\label{def:composed}
Given composed states $\overline{s}, \overline{s}' \in \States^k$, labels $\overline{\ell} \in \Lbl^k$ and a schedule $\varnothing \neq M \subseteq \{1,\ldots,k\}$,
we write $\overline{s}\overset{M,\overline{\ell}}{\rightsquigarrow}\overline{s}'$
to indicate that $\overline{s}'$ is obtained from $\overline{s}$ by taking the transition with label $\ell_i$ from $s_i$ whenever $i \in M$, and stuttering otherwise.
Formally,
\[\overline{s}\overset{M,\overline{\ell}}{\rightsquigarrow}\overline{s}' \iff
 \bigwedge_{i\in M} \TrSem(s_i,\ell_i,s'_i) \land
 \bigwedge_{i\not\in M} s_i = s'_i.\]
Given an arbiter $\arbiter$,
the (semantic, labeled) \emph{composed transition system} according to $\arbiter$ is then
\[
\TSsem^k_\arbiter  =  (\States^k, \Lbl^k, \InitSem^k, \TrSem^k_\arbiter) \qquad \text{ where } \qquad  
\TrSem^k_\arbiter =  \{(\overline{s},\overline{\ell}, \overline{s}') \mid \overline{s} \overset{M,\overline{\ell}}{\rightsquigarrow} \overline{s}' \text{ for } M  \in \arbiter(\overline{s})\}
\]
We refer to $\overline{s}\overset{M,\overline{\ell}}{\rightsquigarrow}\overline{s}'$ as a transition of the composed system with schedule $M$ and label $\overline{\ell}$.
\end{definition}

Each trace of $\TSsem^k_\arbiter$ captures $k$ traces of $\TS$ that are aligned according to $\arbiter$.
As such, considering the composed transition system enables the reduction of $k$-safety verification in $\TS$
to regular safety verification over $\TSsem^k_\arbiter$. The idea is to verify that, for every composed state that is reachable from $\psi$ in $\TSsem^k_\arbiter$, if all individual states are at their observation points, then the composed state satisfies $\phi$.
That is, every such reachable composed state satisfies $(\bigwedge_{i=1}^k \xi_i) \to \phi$.

For this reduction to be sound, the arbiter must ensure that every tuple of $k$ traces that can synchronize on the observation points is considered with some alignment in which all $k$ traces reach the observation points at the same time.
To this end, the arbiter must choose at least one schedule $M$ for each composed state (preventing traces from being ``truncated''). 
Furthermore, the arbiter must respect the synchronization of the observation points: when a trace reaches an observation point, it must wait for all the other traces to reach theirs, before it is allowed to proceed.
This motivates the following definitions.
\begin{definition}[Valid Schedules]
\label{valid-schedules}
$M\in \schedules$ is a \emph{valid schedule} for a composed state $\overline{s} = (s_1,\ldots, s_k)$
if either of the following two conditions holds:
\newcommand\Obs{\mathit{Obs}}

\vspace{2pt}
\begin{inparaenum}
    \item  $\forall i\in M\cdot s_i\not\models\xi_i$\label{lbl:case_1}, or  \qquad
    \item $\forall i\in M\cdot s_i\models\xi_i$ and $M=\{1,...,k\}$.\label{lbl:case_2}
\end{inparaenum}

\vspace{2pt}
\noindent  We denote the set of valid schedules for $\overline{s}$ by $\validS(\overline{s})$.
\end{definition}


For composed traces that are constructed from valid schedules, \Cref{valid-schedules} ensures that
all individual traces reach their observation points before any of them can progress past them (this is enforced by case (\ref{lbl:case_1}), where $M$ is only allowed to include traces that are not in their observations points); and when they do,
they do it simultaneously (this is handled by case (\ref{lbl:case_2}), where all traces are at their observation points, and all are included in the schedule).\footnote{The requirement that all traces leave the observation point in tandem saves the need to record which of them already made a step since the last observation point.}
In other words, the observation points act as a ``barrier'' that all individual traces synchronize on.
To ensure that all composed traces in $\TSsem^k_\arbiter$ are constructed in this way, we then require that a valid arbiter assigns only valid schedules to every reachable composed state:

\begin{definition}[Valid Arbiters]
We say that an arbiter $\arbiter$ is \emph{valid} for $\TS$ and $\varphi$ if for every composed state $\overline{s} \in \States^k$
that is reachable in $\TSsem^k_\arbiter$ from some initial state that satisfies $\psi$, it holds that
$\varnothing\neq \arbiter(\overline{s})
\subseteq \validS(\overline{s})$.
\end{definition}

Similarly to the result shown in~\cite{DBLP:conf/cav/ShemerGSV19} for relational pre/post specifications, if the arbiter is valid, then $\TS \models \varphi$ if and only if the composed transition system $\TSsem^k_\arbiter$ induced by the arbiter satisfies the (regular) safety property ``globally $(\bigwedge_{i=1}^k \xi_i) \to \phi$'', when the initial states of $\TSsem^k_\arbiter$ are restricted to $\{\overline{s} \in \InitSem^k \mid \overline{s} \models \psi\}$.\footnote{Note that for the correctness of the reduction, $\TSsem^k_\arbiter$ need not have a symbolic representation.
In particular, it may be the case that $\arbiter$ is not definable in first-order logic.
} The latter can be established by finding an inductive invariant, namely, 
a set $\InvSem$ of composed states  that includes the initial composed states, is closed under steps of $\TSsem^k_\arbiter$, and is included in the ``safe'' states.
Altogether, we conclude that $\TS \models \varphi$ if and only if there exists a valid arbiter and
a set of composed states $\InvSem$ such that $\InvSem$ is an inductive invariant for the induced composed transition system $\TSsem^k_\arbiter$ that establishes its safety w.r.t.\
$(\bigwedge_{i=1}^k \xi_i) \to \phi$.

\subsection{FOL Encoding of Self Composition}\label{sec:fol-k-safety}
Next, given a transition system $\TS= (\Vone, \Lblone, \Init, \Tr)$, we encode the problem of determining the existence of a valid arbiter and a corresponding inductive invariant as a satisfiability problem in first-order logic.
To reason about composed states, we define a vocabulary $\V  = V_1 \cup \cdots \cup V_k$ that consists
of the set of state variables of all traces.
We encode the arbiter using a family of unknown predicates $\{A_M(\V)\}_{M \in \schedules}$, one for every $\varnothing \neq M \subseteq \{1,\ldots,k\}$,
and encode the inductive invariant using an unknown predicate $\Inv(\V)$.
We express the case where all traces reached an observable state but $\phi$ does not hold using the constraint:
\[\Bad(\V) \eqdef \bigwedge_i \mathit{\xi_i}(V_i) \wedge \neg \phi(\V).\]
We denote the label variables of all the traces by $\Lvoc \eqdef \langle\Lblone_1,\ldots,\Lblone_k\rangle$.
The joint steps of the traces as determined by the schedule $M$ are given by the following constraint:
\[
\begin{array}{rll}
 \Delta_{M}(\V, \V',\Lvoc) 
 & \eqdef &
\bigwedge_{i\in M} \Tr(V_i,a_i,V_i') \wedge \bigwedge_{i\not\in M} V_i = V_i' \\
\delta_{M}(\V, \V') &\eqdef&
\exists \Lvoc
\cdot \Delta_{M}(\V, \V',  
\Lvoc)
\end{array}
\]

$\Delta_{M}$ represents $\overline{s}\overset{M,\overline{\ell}}{\rightsquigarrow}\overline{s}'$ (\cref{def:composed}) in the sense that $\overline{s},\overline{s}',\overline{\ell} \models \Delta_{M}$ if and only if $\overline{s}\overset{M,\overline{\ell}}{\rightsquigarrow}\overline{s}'$.
$\delta_{M}$ existentially quantifies over
the label variables in $\Delta_{M}$,%
\footnote{Since, in our FOL schema, $\delta_M$ appears on the left-hand side of an implication, existential quantifiers can be pushed outside as universal quantifiers, resulting in quantifier-free bodies.} indicating that any labeled transition can be used to make a step.
(The labels will become important in the next section.)
\iflong

\else \fi
The definition of a valid schedule is captured by:
\begin{equation*} 
\valid_M(\V) ~\eqdef~
\begin{cases}
\bigwedge_{i \in M} \neg \xi_i(V_i) & M \neq \{1,\ldots,k\}\\
\big(\bigwedge_{i \in M} \neg \xi_i(V_i)\big) \vee \big(\bigwedge_{i \in M}  \xi_i(V_i)\big) & M = \{1,\ldots,k\}\\
\end{cases}
\end{equation*}
This definition ensures that $M \in \validS(\overline{s})$ (see \Cref{valid-schedules}) if and only if $\overline{s} \models \valid_M(\V)$.

\begin{figure}[t]
\resizebox{1\textwidth}{!}{
$
\renewcommand{\arraystretch}{1.2}
\begin{array}{l|l}
\subfloat[]{$
\begin{array}{lr@{}l}
    & \bigwedge_i \Init(V_i) \land \psi(\V) \rightarrow{} & \Inv(\V)  \\
    & \Inv(\V) \land \Bad(\V) \rightarrow{} & \bot \\
    \eqForall & \Inv(\V) \wedge A_{M}(\V) \land \neg \valid_M(\V) \rightarrow{} & \bot \\
    \eqForall & \Inv(\V)\land A_{M}(\V) \land \delta_M(\V, \V')
 \rightarrow{} & \Inv(\V') \\
    & \Inv(\V) \rightarrow{} & \displaystyle \bigvee_{M} A_{M}(\V)
\end{array}
$
\label{safety:formulas:arbiter}
}
&
\subfloat[]{$
\begin{array}{lr@{}l}
    & \displaystyle\bigwedge_M D_M(\V) \wedge \bigwedge_i \Init(V_i) \wedge \psi(\V) \rightarrow{} & \bot  \\
    \eqForall &  \Bad(\V) \rightarrow{} & D_M(\V) \\
    \eqForall & \neg \valid_M(\V) \rightarrow{} & D_M(\V) \\
    \eqForall & \displaystyle \bigwedge_{M'} D_{M'}(\V') \land \delta_M(\V, \V')
 \rightarrow{} & D_M(\V)
\end{array}
$
\label{safety:formulas:doomed}
}
\end{array}
$
}

\vspace{-1.3em}
\raggedright\scalebox{0.8}{$\big(\,\eqForall = \forall M\in\schedules\big)$}\hspace{3em}
\caption{$k$-safety verification scheme before (a) and after (b) the transformation.}
\label{safety:formulas}
\end{figure}

\Cref{safety:formulas:arbiter}
formalizes the joint requirements of the arbiter and the inductive invariant that ensure that $\varphi$ holds in $\TS$.
The formulas in lines 1,2 and 4 ensure that $\Inv$ is an inductive invariant that establishes safety of $\TSsem^k_\arbiter$ w.r.t. the regular safety property $\neg \Bad \equiv (\bigwedge_{i=1}^k \xi_i) \to \phi$, the formulas in line 3 ensure that $\arbiter$ is a valid arbiter and the formula in line 5 ensures that $\arbiter$ determines at least one schedule for every reachable state.
The following theorem summarizes the soundness of the encoding, which is a slight generalization of the encoding in~\cite{DBLP:conf/cav/ShemerGSV19} (where only pre/post specifications are considered):
\begin{theorem}\label{thm:k-safety}
The set of formulas in \Cref{safety:formulas:arbiter} is satisfiable iff $\TS \models \varphi$.
\end{theorem}

\begin{proof}[Proof sketch]
The proof is similar to the proof of~\cite{DBLP:conf/cav/ShemerGSV19}. The differences are in (1)~the notion of valid schedules, which ensures that no tuple of $k$ traces that can reach an observation point simultaneously is overlooked, much like the notion of fairness from~\cite{DBLP:conf/cav/ShemerGSV19} ensures that no tuple of terminating $k$ traces is overlooked, and (2)~the definition of bad states as states that satisfy $\bigwedge_i \mathit{\xi_i}(V_i) \wedge \neg \phi(\V)$, which generalizes their definition in~\cite{DBLP:conf/cav/ShemerGSV19} as terminating states that falsify the post condition.
\end{proof}

\begin{example}
Applying the scheme of \Cref{safety:formulas:arbiter} to the program and  \aehyper specification  of the $2$-safety property from \Cref{intro:running-example} results in \Cref{intro:running-example:arbiter}, 
\SH{BUG fix:}except for moving constraints to the right hand side of the implication when it assists readability.
Note that in this example, the observation points $\xi_i$ of both traces correspond to the condition for exiting the loop (which is the negated loop condition). As a result $\valid_{\{i\}} \eqdef a_i < b_i$ for $i \in \{1,2\}$ and $\valid_{\{1,2\}} \eqdef (a_1 < b_1 \wedge a_2 < b_2) \lor (\neg(a_1 < b_1) \land \neg (a_2 < b_2))$. 
\end{example}


\subsection{From FOL to CHC Encoding of Self Composition}
The set of formulas in \Cref{safety:formulas:arbiter} fits the general scheme of \Cref{transformation:formulas:arbiter};
Thus, it is amenable to our general satisfiability-preserving transformation, resulting in
the CHCs in \Cref{safety:formulas:doomed}.
%
%
%
%
Since the transformation is satisfiability preserving, we obtain the following as a corollary of \Cref{thm:k-safety,thm:trans}:
\begin{corollary}
The system of CHCs in
{\upshape\Cref{safety:formulas:doomed}}
is satisfiable iff $\TS \models \varphi$.
\end{corollary}

Where $A_M(\V)$ in \Cref{safety:formulas:arbiter} describes the states where choosing schedule $M$ leads to successful verification
with $\Inv$ as an inductive invariant, $\D_M(\V)$ in \Cref{safety:formulas:doomed} can be understood as describing states where choosing $M$ would \emph{prevent} the verification from going through in the sense that no inductive invariant would exist.
In other words, these states are ``doomed'' if $M$ is chosen, hence the choice of notation.
If the set of CHCs in \Cref{safety:formulas:doomed} is satisfiable,
it proves that initial states that satisfy the pre-condition are not doomed. This intuition
can be interpreted in a dual manner: if the initial states are not doomed, then there
exists an alignment for which a safe inductive invariant exist.

%% file: 06-beyond.tex
\section{Encoding $\forall^*\exists^*$ Hyperproperties as CHCs}
\label{sec:beyond}

\begin{figure}[t]
\lstset{basicstyle=\ttfamily\scriptsize}
\begin{tabular}{@{\hspace{1cm}}l@{\hspace{1cm}}l}
\begin{lstlisting}[numbers=left]
sum = 0;
b = *;
if (b > 0) {
  i = 0;
  while (i < n - 1) {
    sum = sum + A[i];
    i++;
  }
}
\end{lstlisting}
&
\begin{lstlisting}[numbers=left,firstnumber=10]
else {
  i = 1;
  while (i < n) {
    y = *;     
    sum = sum + A[i] + y;
    i++;
  }
}
\end{lstlisting}
\end{tabular}

\vspace{2pt}
$(A_1 = A_2 \land n_1 = n_2) \rightarrow \forall\pi_1:\pc=5~\exists\pi_2:\pc=5\lor\pc=12\cdot
    \square(b_2 \leq 0 \land \mathit{sum}_1=\mathit{sum}_2)$

\vspace{-0.25em}
\caption{Example for a $\forall\exists$ hyperproperty.}
\label{beyond:example}
\end{figure}

In this section we consider the more general case of \aehyper specifications. 
Throughout the section, $\TS$ is a transition system, and we fix a formula:
%
\[\varphi ~=~ \psi \rightarrow \forall \pi_1 : \xi_1, \ldots, \pi_l : \xi_l\cdot \exists \pi_{l+1} : \xi_{l+1},\ldots,\pi_k:\xi_k\cdot~ \square\phi
\]

In order to encode the problem of deciding if $\TS\models\varphi$ as
a satisfiability problem, we follow \cite{DBLP:conf/cav/BeutnerF22}, and consider a \emph{game semantics}, which
is natural due to the alternation of quantifiers. 
The $\forall$ quantifiers are ``demonic'' and hence controlled by the falsifier, and the $\exists$ quantifiers are ``angelic'', thus controlled by the verifier.


In the following, we introduce the  game semantics  of~\cite{DBLP:conf/cav/BeutnerF22} for \aehyper. We then 
encode truth of $\varphi$ in $\TS$ under the game semantics as a satisfiability problem, and use the transformation from \Cref{transformation} 
to obtain a system of CHCs that is satisfiable iff $\TS$ satisfies $\varphi$ according to the game semantics.


%

\subsection{Game Semantics for \aehyper}

The game semantics for \aehyper proposed by~\cite{DBLP:conf/cav/BeutnerF22} introduces a safety game between a verifier and a falsifier. 
Truth of $\varphi$ in $\TS$ under the game semantics is captured by the verifier having a winning strategy in the game for $\TS$ and $\varphi$.


\begin{example}
\label{beyond:example:desc}
To illustrate the game semantics, 
we use the example in \Cref{beyond:example}, which accompanies this section.
The presented program computes the sum of an array slice, nondeterministically choosing between the slice $A[0..n-2]$ and $A[1..n-1]$.
For the second case, an arbitrary integer can be added to each summand.
This allows the program to fulfill the specification at the bottom,
which requires that for every execution there is a corresponding execution of the second case ($b_2 \leq 0$) such that the sums at lines 5 and 12 align at every iteration.
The specification is valid because $y$ at line 13 can always be chosen to compensate for the deviation between $A[i_1]$ and $A[i_2]$. 

Considering the game semantics, the falsifier chooses $b_1$, and the verifier must choose $b_2\leq 0$ to satisfy the specification,
a scheduling that will align $\pc_1=5$ and $\pc_2=12$
at every iteration, and 
a value for $y$
such that after both assignments (lines 6 and 14) $\mathit{sum}_1=\mathit{sum}_2$ is satisfied.
Setting $y=A[i_1] - A[i_2]$ achieves this objective.
\end{example}



\paragraph{Safety games} 
\iflong A \emph{safety game} is a game
\else 
are 
\fi 
played between a \emph{verifier}, whose goal is to avoid \emph{bad} states, and a \emph{falsifier} who tries to reach a bad state. Formally, the game is a tuple $\game = (\vstates, \fstates, \istates, \vsteps, \fsteps, \bstates)$ where
$\vstates$ are \emph{verifier states}, in which the verifier moves,
and $\fstates$ are \emph{falsifier states}, in which the falsifier moves, and $\vstates \cap \fstates = \varnothing$.
The \emph{game states} are $\gstates = \vstates \cup \fstates$. $\istates \subseteq \gstates$ is a set of initial states, and $\bstates \subseteq \gstates$ is a set of bad states. $\vsteps \subseteq \vstates \times \gstates$ defines the possible moves of the verifier and $\fsteps \subseteq \fstates\times\gstates$---of the falsifier. It is assumed that $\vsteps,\fsteps$ are total 
\iflong
in the sense that there is always at least one move for each player in the corresponding states.
\else 
\ie,  there is at least one move for each player from every state.
\fi
A \emph{play} is a  sequence of game states
$\sigma_0,\sigma_1,\ldots$ such that $\sigma_0 \in \istates$,
and for every $i \geq 0$,  $(\sigma_i,\sigma_{i+1}) \in\vsteps  \cup \fsteps$.
The play is winning for the verifier if it is infinite and $\sigma_i \not \in \bstates$ for every $i \geq 0$.
A (memoryless) strategy for the verifier is a function $\strategy: \vstates \to \gstates$ 
such that $(\sigma,\strategy(\sigma))\in \vsteps$ for every $\sigma \in \vstates$. $\strategy$ is a \emph{winning strategy} for the verifier if all the plays in which the verifier moves  according to $\strategy$ are winning for the verifier.



\paragraph{Game semantics for $\forall^*\exists^*$-OHyperLTL}
%
%
Let $\varphi$ be as above.
The game that captures the semantics of $\varphi$ is defined with respect to a deterministic labeled transition system 
$\TS=(V, \Lblone ,\Init, \Tr)$.
Note that the assumption that $\TS$ is deterministic does not restrict generality, and, in particular, does not prevent treatment of nondeterministic programs, 
since we can always determinize $\TS$ by extending the set of labels
without affecting the semantics; this step may introduce infinitely many labels, 
which do not require any special treatment in the definition of the game, but whose CHC encoding will be addressed in \Cref{sec:beyond:infinite}.
From this point on, we assume that $\TS$ is deterministic.

%
The game for $\varphi$ and $\TS$ proceeds in rounds, where in each round the falsifier makes a move and the verifier responds.
The falsifier states 
are composed states (of $k$ traces), and the verifier states augment them with a record of the falsifier's last move. 
The bad states are falsifier states where all traces are in their observation points but $\phi$ does not hold.
The falsifier is responsible for choosing the transitions that define the $\forall$ traces
\iflong $t_1,\ldots,t_l$
\else $t_{1..l}$
\fi
assigned to 
\iflong $\pi_1,\ldots,\pi_l$. 
\else 
$\pi_{1..l}$.
\fi
The verifier responds by choosing the transitions of the $\exists$ traces 
\iflong $t_{l+1},\ldots,t_k$ 
\else $t_{l+1..k}$
\fi 
assigned to 
\iflong 
$\pi_{l+1},\ldots,\pi_k$.
\else 
$\pi_{l+1..k}$.
\fi
Here the labels of the transitions come into play: the players specify the transitions of choice by picking a label $\ell \in\Lbl$ for each trace \iflong $t_i$.
\else.
\fi 
(Since $\TS$ is deterministic, transitions are uniquely identified by labels.)
The traces 
\iflong $t_1,\ldots,t_k$ 
\fi 
then need to be aligned s.t.\ they synchronize on their observation points defined by $\xi_i$. 
As long as the alignment of the traces is valid, i.e., ensures that the traces reach their observation points simultaneously when possible, the alignment does not affect the winner of the play. 
That is, if a play is winning for the verifier with one (valid) alignment, it will also be winning with all others.
However, as in the case of $k$-safety, the alignment  is instrumental for obtaining a winning strategy that has a simple description. As a result, the choice of the (valid) alignment is also left to the verifier, and is defined using (valid) schedules.
Altogether, 
a move of the falsifier consists of picking labels $\ell_{1},\ldots,\ell_l\in\Lbl$ for the 
\iflong universally quantified 
\else $\forall$ 
\fi 
trace variables;
a move of the
verifier consists of picking a valid subset
$\varnothing\neq M\subseteq\{1,\ldots,k\}$
of the traces to progress (as in \Cref{safety}) and also labels $\ell_{l+1},\ldots,\ell_k\in\Lbl$ for the 
\iflong existentially quantified
\else $\exists$-%
\fi
trace variables, and proceeding to the resulting composed state.\iflong\footnote{In our definition of the game for $\TS$ and $\varphi$, each round consists of two steps as opposed to 3 in~\cite{DBLP:conf/cav/BeutnerF22}. This definition is more precise than the def of~\cite{DBLP:conf/cav/BeutnerF22} in the following sense: a winning strategy in the game of~\cite{DBLP:conf/cav/BeutnerF22} implies a winning strategy in our game.}
\else\footnote{In~\cite{DBLP:conf/cav/BeutnerF22}, steps of the verifier are split to two. Our definition is more precise in the sense that a winning strategy in the game of~\cite{DBLP:conf/cav/BeutnerF22} implies a winning strategy in our game.}
\fi
In this manner, the verifier iteratively ``reads off'' the labels chosen by the falsifier for $t_{1..l}$, and generates the labels for the traces $t_{l+1..k}$, properly aligned by its choice of $M$, while avoiding the bad states.
If the verifier can do so indefinitely, then this proves that $\varphi$ holds.
%
The formal definition of the components of the game follows.

\begin{definition}
Let $\TS$  be a deterministic transition system and $\TSsem = (\States, \Lbl,\InitSem, \TrSem)$ its semantic counterpart.
The verification game for $\TS$ and $\varphi$ is a safety game 
$\game_{\TS,\varphi} = (\vstates, \fstates, \istates, \vsteps, \fsteps, \bstates)$, where: 

\vspace{5pt}
$
\begin{array}{@{}l@{\hspace{0.25em}}l@{}}
\fstates & = \States^k \qquad
\vstates = \States^k \times \Lbl^l \qquad
\istates = \{\overline{s} \in \InitSem^k \mid \overline{s} \models \psi\}\\[.3em]
\bstates &= \{\overline{s} \in \fstates \mid \overline{s} \not \models \phi \text{ and }s_i \models \xi_i \text{ for every } 1 \leq i \leq k\}\\
\fsteps & =  \{(\overline{s}, (\overline{s},\overline{\ell}{}^\forall)) \mid \overline{s} \in \fstates, ~\overline{\ell}{}^\forall \in \Lbl^l\} \quad
\vsteps =
\{((\overline{s},\overline{\ell}{}^\forall), \overline{s}') \mid \overline{s} \overset{M,\overline{\ell}}{\rightsquigarrow} \overline{s}' \text{ for } M \in \validS(\overline{s}) \text{ and } \overline{\ell}{}^\exists
\in\Lbl^{k-l}
\}
\end{array}
$

\SH{should we write the steps of the players as a union of sets of steps, a set for each possible choice? (each set would be a singleton....) It can help explain the encoding later}In the above, $M$ represents a valid schedule according to \Cref{valid-schedules}
and $\overline{s}\overset{M,\overline{\ell}}{\rightsquigarrow}\overline{s}'$ denotes
a transition of the composed system from $\overline{s}$ to $\overline{s}'$
according to schedule $M$ and the labels in $\overline{\ell}$
 (see \Cref{def:composed}).
The labels $\overline{\ell} = \langle\ell_1,\ldots,\ell_k\rangle$ are split into $\overline{\ell}{}^\forall=\langle\ell_1,..,\ell_l\rangle$
and $\overline{\ell}{}^\exists=\langle\ell_{l+1},..,\ell_{k}\rangle$.

\end{definition}
\iflong 
For a falsifier step $(\overline{s}, (\overline{s},\overline{\ell}{}^\forall)) \in \fsteps$,
we refer to $\overline{\ell}{}^\forall$ as the choice of the falsifier (note that $\overline{s}$ does not change).
Similarly, for a verifier step $((\overline{s},\overline{\ell}{}^\forall), \overline{s}')\in \vsteps$,
we refer to $\langle M,\overline{\ell}{}^\exists \rangle$ such that $\overline{s} \overset{M,\overline{\ell}}{\rightsquigarrow} \overline{s}'$
as the choice of the verifier that leads to $\overline{s}'$.
\fi

\begin{example}
In the example of \Cref{beyond:example}, the labels of transitions are integer values 
that reflect the choice of \code{*}
at lines 2 and 13 (and have no effect on other states).
The verifier and falsifier specify their moves using these labels. 
For example, in order to ensure that 
$\mathit{sum}_1=\mathit{sum}_2$ is satisfied at every iteration, the verifier selects a transition label $\ell = A[i-1] - A[i]$ in line 13, which sets the value of $y$ accordingly;
after both assignments at lines 6 and 14, $\mathit{sum}_1=\mathit{sum}_2$ holds. 
\end{example}

The game semantics of \aehyper is based on the winner in the verification game:
\begin{definition}[Game Semantics for \aehyper~\cite{DBLP:conf/cav/BeutnerF22}]
Let $\TS$ be a transition system and $\varphi$ a \aehyper formula.
$\TS$ \emph{satisfies} $\varphi$ \emph{according to the game semantics}, denoted $\TS \modelsg\varphi$, if the verifier has a winning strategy in the verification game
$\game_{\TS,\varphi}$.
\end{definition}


As shown in~\cite{DBLP:conf/cav/BeutnerF22}, the game semantics is sound but incomplete. Incompleteness means that it is possible that $\TS \models \varphi$ but the verifier does not have a winning strategy in the safety game for $\varphi$ and $\TS$.
Soundness is summarized by the following theorem:


\begin{theorem}[
{\rm\cite{DBLP:conf/cav/BeutnerF22}}]
If $\TS \modelsg \varphi$ then $\TS \models \varphi$.
\end{theorem}

\iflong
We note that with $\forall^*\exists^*$ properties, the existence of an infinite trace $t_i$, $l+1\leq i\leq k$ in which $\xi_i$ is never satisfied means that the property holds vacuously.
In this case, a winning strategy for the verifier is also trivial: to always select that trace and only progress along it, effectively starving all the other traces.
Such situations are reminiscent of pre- and postconditions in safety verification, which do not address cases of nontermination.
Interesting cases of hyperproperties 
are such where $\xi_i$ is enabled infinitely often in all traces. 
\fi


\subsection{CHC Encoding of \aehyper Verification in the Case of Finite Branching}
\label{sec:beyond:finite}
Having defined the game semantics based on~\cite{DBLP:conf/cav/BeutnerF22}, we now 
encode the existence of a winning strategy for the verifier as a satisfiability problem in FOL, and use the transformation from \Cref{transformation} 
to obtain a system of CHCs. 
We start with the case of transition systems with finite branching, in which the CHC encoding is both sound and complete w.r.t. the game semantics. 
That is, the set of CHCs is satisfiable if and only if there exists a winning strategy for the verifier.
In the next subsection, we handle infinite branching, where completeness is no longer guaranteed (both for the FOL encoding and for the CHC encoding), but soundness is preserved. 

To encode the game semantics of $\varphi$ in $\TS$,
we introduce unknown predicates $\{A_\au\}_{\au\in\U}$ that describe a strategy for the verifier in $\game_{\TS,\varphi}$, where $\U$ is the set of possible choices of the verifier. 
In addition, we introduce an unknown predicate $\Inv$
that encodes an inductive invariant that ensures  the strategy is winning \iflong (i.e., when playing by the strategy the verifier never reaches a bad state). \else. \fi

Recall that we first consider the case where $\TS$ has a finite branching degree, i.e., the set of labels $\Lbl$ is finite.
This makes it possible to define $\U$ as the set of all possible concrete choices of the verifier and
introduce a predicate $A_\au$ per every possible choice of the verifier.
To do so, we define $\U=\allScheds\times\Lbl^{k-l}$, where $\allScheds = \mathcal{P}(\{1,\ldots,k\})\setminus \{\varnothing\}$ is the set of possible schedules
(see \Cref{def:arbiter}), and $\Lbl^{k-l}$ are the choice labels for constructing the traces assigned to $\{\pi_i\}_{i=l+1..k}$. Note that $\U$ is finite in this case.
For each $\au = \langle M,\overline{\ell}{}^\exists \rangle\in\U$, the predicate $A_\au$
describes the verifier states in which the verifier chooses $\au$ for its move.
Recall that verifier states consist of both the previous state of the verifier,
captured by the composed state vocabulary $\V$ defined as before,
and the last move of the falsifier, captured by 
label variables $\langle\Lblone_1,\ldots,\Lblone_l\rangle$.
Recall that $\Lvoc = \langle\Lblone_1,\ldots,\Lblone_k\rangle$. 
We denote $\LvocA = \langle\Lblone_1,\ldots,\Lblone_l\rangle$,
$\LvocE = \langle\Lblone_{l+1},\ldots,\Lblone_k\rangle$, such that
$\Lvoc =  \LvocA \cup \LvocE$.
Then, the $A_\au$ predicates are defined over $\V \cup \LvocA$.
The $\Inv$ predicate 
is defined over $\V$ only, as it describes a set of falsifier states.

The formulas in \Cref{beyond:formulas:arbiter}
formalize the requirements that ensure that
$\{A_\au\}_\au$ defines a winning strategy for the verifier, in which the set of reachable falsifier states are overapproximated by $\Inv$.
The formulas are similar to the ones in \Cref{safety:formulas:arbiter}. In particular, the bad states and valid schedules are encoded by $\Bad(\V)$ and $\valid_M(\V)$, respectively, defined in \Cref{sec:fol-k-safety}.
To account for the alternating choices of the falsifier
($\overline{\ell}{}^\forall$) and verifier ($\langle M, \overline{\ell}{}^\exists\rangle$) in every round,
we define
%
\[
\begin{array}{lcl}
 \Delta_{M}(\V, \V', \Lvoc) & \eqdef &
\bigwedge_{i\in M} \Tr(V_i,a_i,V_i') \wedge \bigwedge_{i\not\in M} V_i = V_i' \\
\delta_{M,\overline{\ell}{}^\exists}(\V, \V', \LvocA) & \eqdef &
\Delta_{M}(\V, \V', \Lvoc)\big[\LvocE \mapsto \overline{\ell}{}^\exists\big]
\end{array}
\]
$\Delta_M$ is the formula defined in \Cref{sec:fol-k-safety} to capture the semantics of a transition in a composed system, \ie, 
$\overset{M, \overline{\ell}}{\rightsquigarrow}$ (see \Cref{def:composed}).
$\delta_{M,\overline{\ell}{}^\exists}$ is then the projection of $\Delta_M$
to a concrete choice of labels $\overline{\ell}{}^\exists$ for the existentially quantified traces;
the labels for the universals, captured by $\LvocA$, remain free.
\iflong 
This means that
$\overline{s}, \overline{s}', \overline{\ell}{}^\forall$ (valuations to $\V,\V',\LvocA$)
satisfy $\delta_{M,\overline{\ell}{}^\exists}$ iff
verifier choice 
$\langle M,\overline{\ell}{}^\exists \rangle$  from $(\overline{s},\overline{\ell}{}^\forall)$ leads to $\overline{s}'$.
The formulas in line 4 of \Cref{beyond:formulas:arbiter} then ensure that if state $\overline{s}'$ is reached by step $\overline{\ell}{}^\forall$ of the falsifier followed by step $\langle M,\overline{\ell}{}^\exists \rangle$ of the verifier from $\overline{s} \models \Inv$ s.t.\ $\overline{s},\overline{\ell}{}^\forall \models A_{M,\overline{\ell}{}^\exists}$, then $\overline{s}' \models \Inv$. 

\fi








\iflong
The verifier having a winning strategy in the game corresponds to satisfiability of the  formulas in \Cref{beyond:formulas:arbiter}, which ensures soundness and completeness of the FOL encoding w.r.t.\ the game semantics:
\fi
%

\begin{figure}[t]
\hspace{-0.05\textwidth}
\resizebox{1.1\textwidth}{!}{
$
\renewcommand{\arraystretch}{1.2}
\begin{array}{l|l}
\subfloat[]{$
\begin{array}{lr@{}l}
    & \bigwedge_i \Init(V_i) \land \psi(\V)  \rightarrow{} & \Inv(\V)  \\
    & \Inv(\V) \land \Bad(\V) \rightarrow{} & \bot \\
    \eqForall & \Inv(\V) \wedge A_{M,\overline{\ell}{}^\exists}(\V,\LvocA) \land
\lnot\valid_M(\V) \rightarrow{} & \bot \\
    \eqForall & \Inv(\V)\land A_{M,\overline{\ell}{}^\exists}(\V,\LvocA) \land \delta_{M,\overline{\ell}{}^\exists}(\V, \V', \LvocA)
 \rightarrow{} & \Inv(\V') \\[.3em]
\multicolumn{3}{r}{
    \Inv(\V) \rightarrow \displaystyle \hspace{-1em}\bigvee_{\langle M,\overline{\ell}{}^\exists\rangle\in\U} \hspace{-1em}
    A_{M,\overline{\ell}{}^\exists}(\V,\LvocA)}
\end{array}
$
\label{beyond:formulas:arbiter}
}
&
\subfloat[]{$
\begin{array}{l@{\hspace{-0.5em}}r@{}l}
    & \displaystyle\bigwedge_{\langle M,\overline{\ell}{}^\exists\rangle\in\U} D_{ M,\overline{\ell}{}^\exists}(\V) \wedge \bigwedge_i \Init(V_i) \wedge \psi(\V) \rightarrow{} & \bot  \\
    \eqForall &  \Bad(\V) \rightarrow{} & D_{ M,\overline{\ell}{}^\exists}(\V) \\
    \eqForall & \neg \valid_M(\V) \rightarrow{} & D_{ M,\overline{\ell}{}^\exists}(\V) \\
    \eqForall & \displaystyle \bigwedge_{\langle M',\overline{\ell}'{}^\exists\rangle\in\U} D_{ M',\overline{\ell}'{}^\exists}(\V') \land \delta_{ M,\overline{\ell}{}^\exists}(\V, \V')
 \rightarrow{} & D_{ M,\overline{\ell}{}^\exists}(\V)
\end{array}
$ \label{beyond:formulas:doomed}
}
\end{array}
$
}
\vspace{-0.6cm}

\raggedright\scalebox{0.8}{$\big(\,\eqForall = \forall \langle M,\overline{\ell}{}^\exists\rangle\in\U\big)$}
\caption{\aehyper verification scheme before (a) and after (b) the transformation.}
\label{beyond:formulas}
\vspace{-1em}
\end{figure}

\begin{theorem} \label{thm:game-hyper-fol-encoding}
The set of formulas in \Cref{beyond:formulas:arbiter} is satisfiable iff $\TS \modelsg \varphi$.
\end{theorem}
\begin{proof}
A solution for \Cref{beyond:formulas:arbiter} induces a winning strategy $\chi$ for the verifier in the
game for $\varphi$ and $\TS$:
\iflong the strategy $\strategy$ is given by 
\fi 
$\chi(\overline{s},\overline{\ell}{}^\forall) = \overline{s}'$ for $\overline{s} \models \Inv$,
where $\overline{s}'$ is reached by choosing $\langle M, \overline{\ell}{}^\exists \rangle$
(\ie, $\overline{s}, \overline{s}', \overline{\ell}{}^\forall  \models \delta_{M,\overline{\ell}{}^\exists}$) such that $\overline{s},\overline{\ell}{}^\forall \models A_{M,\overline{\ell}{}^\exists}$;
such $\overline{s}'$ must exist because 
the last formula
states that there must always be a choice for the verifier in falsifier states that satisfy $\Inv$.
For $\overline{s} \not\models \Inv$, $\chi(\overline{s},\overline{\ell}{}^\forall)$ is defined arbitrarily.
In the other direction, given a winning strategy for the verifier, we define the interpretation of $\Inv$ to be its winning region and the interpretation of $A_{M, \overline{\ell}{}^\exists}$ to consist of the falsifier states ($\overline{s},\overline{\ell}{}^\forall$) where the strategy chooses $\overline{s}'$  
such that $\overline{s},\overline{\ell}{}^\forall \models A_{M,\overline{\ell}{}^\exists}$.
\end{proof}

\begin{remark}\label{rem:safety-as-aehyper}
For $k$-safety properties, the encoding in \Cref{beyond:formulas:arbiter}, based on the games semantics, is equivalent to the encoding in \Cref{safety:formulas:arbiter} (\Cref{safety}).
In particular, in this case, the set $\LvocE$ is empty, which means that $\overline{\ell}{}^\exists = \langle \rangle$, 
resulting in a game with finite branching, namely only the choices of the schedule $M$.
Note that for such properties, the benefits of the game semantics are less obvious since if $\TS \models \varphi$, then \emph{every} strategy is winning for the verifier.
\end{remark}

\medskip
Applying our transformation to the formulas in \Cref{beyond:formulas:arbiter} results in the \iflong system of \fi  CHCs in \Cref{beyond:formulas:doomed}.
Intuitively, $A_{M,\overline{\ell}{}^\exists}$ 
\iflong (in \Cref{beyond:formulas:arbiter}) \fi
describe the winning strategy for the verifier:
for ``safe'' states, represented by $\Inv$, and given a move made by the falsifier, if the verifier chooses to move
according to $\langle M,\overline{\ell}{}^\exists\rangle$, then it stays in the ``safe'' region. 
In contrast, \iflong the uninterpreted predicate 
\fi
$D_{ M,\overline{\ell}{}^\exists}$ represents ``doomed'' states. 
Namely, if the verifier chooses to move according to $\langle M,\overline{\ell}{}^\exists\rangle$ from a 
state in $D_{ M,\overline{\ell}{}^\exists}$, then the falsifier can force reaching a bad state for every choice of the verifier  in the next steps of the game.
\iflong Due to that, these states are referred to as ``doomed'', as there is no winning strategy for the verifier if a wrong choice is made.
\fi

\begin{corollary}
\label{cor:chc-encoding-game-finite}
The set of CHCs in \Cref{beyond:formulas:doomed} is satisfiable iff $\TS \modelsg \varphi$.
\end{corollary}

\begin{example}
The example in  \Cref{beyond:example} fits the case of finite branching if we assume that the integer values in the array $A$ and those of $\mathit{sum}$ and $y$ are bounded modulo $2^m$, and so are the labels $\Lbl$.
This means that the falsifier has $2^m$ possible steps at each game state, and the verifier has $3\cdot 2^m$ ($3$ is the number of possible schedules out of $\{1,2\}$).
In the next subsection we explain how to encode the problem when the integers are considered to be unbounded. 
\end{example}

\subsection{CHC Encoding of \aehyper Verification in the Case of Infinite Branching}
\label{sec:beyond:infinite}

The set of formulas in \Cref{beyond:formulas:arbiter}, and the corresponding system of CHCs in \Cref{beyond:formulas:doomed}
is well defined when the set $U$ is finite.
However, if $\Lbl$ is infinite, i.e. $\TS$ has infinite branching, so is $U$.
In this case, instead of using $\Lbl^{k-l}$ to specify the traces chosen by the verifier,
we define a finite, abstract set of composed labels, denoted $\ALbl$,
to be used by the verifier (the falsifier continues to use the concrete labels to specify their transitions of choice).
Each abstract label in $\ALbl$ is a relational predicate $p$ with free variables $\V$
(the composed vocabulary) that relates the states of different traces.
%
Thus, the vector of individual existential choices $\overline{\ell}{}^\exists$ of the verifier is now replaced with a \emph{single} choice of a (relational) predicate $p\in \ALbl$ over all the copies.
Intuitively, unlike\SH{improve?}
the use of concrete labels to specify the (unique) next transition for each trace 
individually, 
an abstract label $p\in \ALbl$ determines 
the  next transitions 
for the $\exists$  traces 
by restricting the \SH{slight shortening here:}resulting composed post-states. 

%

Specifically, given a composed state $\overline{s}$, a set of labels $\overline{\ell}{}^\forall$ for the $\forall$ traces
and a schedule $M$, 
a predicate $p \in \ALbl$
is used as a \emph{restriction} (inspired by the homonymous concept from \cite{DBLP:conf/cav/BeutnerF22})
of the transitions of the composed system 
according to schedule $M$ with $\forall$-choices 
$\overline{\ell}{}^\forall$,
restricting the set of transitions to those whose target states
satisfy $p$: 
\[
\restrict_{M,p}(\overline{s},\overline{\ell}{}^\forall) \eqdef \{(\overline{s},\overline{\ell}, \overline{s}') \mid \overline{s}' \models p \land \overline{s} \overset{M,\overline{\ell}}{\rightsquigarrow} \overline{s}' \text{ for some } \overline{\ell}{}^\exists \text{ s.t.\ } \overline{\ell} = \overline{\ell}{}^\forall \overline{\ell}{}^\exists\}.
\]

\begin{example}
In \Cref{beyond:example}, at line 13, a  nondeterministic integer value is assigned to variable $y$.
Since the set of integers is infinite, assigning a  unique label $\ell$ to each integer results in an infinite set $\Lbl$.
To specify the choices of the verifier, we therefore define a finite set of abstract labels. An example of such a set is
$\ALbl = \{\mathit{sum}_1 = \mathit{sum}_2,
  \mathit{sum}_1 = y_2,
  \mathit{sum}_1 < y_2,
  \mathit{sum}_1 = \mathit{sum}_2 + A_2[i_2] + y_2\}$. 
%
The restriction $\mathit{sum}_1 = \mathit{sum}_2$
can result in an empty set of transitions (we will return to this point later in the section); 
but the restrictions $\mathit{sum}_1 = y_2$, $\mathit{sum}_1 < y_2$ and 
$\mathit{sum}_1 = \mathit{sum}_2 + A_2[i_2] + y_2$
always define a nonempty set of transitions when $\pc_2 = 13$ and when a schedule $\{2\}\subseteq M$ is chosen: 
those transitions that choose a value for $y_2$ such that the predicate holds after the transition;
there is always at least one such value. 
In fact, for $\mathit{sum}_1 = y_2$ and $\mathit{sum}_1 = \mathit{sum}_2 + A_2[i_2] + y_2$ there is exactly one such value,
while for $\mathit{sum}_1 < y_2$, the set of values (transitions) is infinite.
Note that there are transitions that are not selected by any restriction (those that assign to $y_2$ a value such that none of the predicates hold).
\end{example}

Thus, the abstract labels define a space of \emph{underapproximations} of the transitions of the composed system---%
each of them includes a subset of the actual composed transitions. 



As before, the falsifier uses $\overline{\ell}{}^\forall$ to specify its choice of transitions for the
traces assigned to the universally quantified variables $\pi_{1..l}$. 
The verifier, on the other hand, uses $p \in \ALbl$ 
to specify the transitions 
for the existentially quantified variables $\pi_{l+1..k}$.
%
We then require that \emph{all} of the composed post-states $\overline{s}'$
reached by the verifier's choice $\langle M,p \rangle$ in response to the falsifier's choice 
$\overline{\ell}{}^\forall$ from $\overline{s}$
are winning for the verifier.
This amounts to  proving that \emph{all} restricted traces satisfy $\square\phi$, which would mean that there \emph{exist} traces that do, \emph{as long as the restrictions do not lead to an empty set of traces}.
Therefore, to ensure soundness of the encoding, we require that the restrictions applied to 
$(\overline{s},\overline{\ell}{}^\forall)$
be nonempty. (Note that the choices of the falsifier, $\overline{\ell}{}^\forall$, are not restricted.) 
%

Rather than limiting the set of predicates $p$ used as abstract labels for $(\overline{s},\overline{\ell}{}^\forall)$ to ensure that $\restrict_{M,p}(\overline{s},\overline{\ell}{}^\forall)$ is nonempty,
we ensure nonemptiness by 
applying a restriction to $(\overline{s},\overline{\ell}{}^\forall)$ only when the resulting set of transitions is nonempty;
otherwise, the full set of transitions  is considered. Note that the full set of transitions is never empty since the transition relation is total.
This gives rise to the following definition of a restricted set of transitions according to $p$:
\begin{equation*}
\restricttag_{M,p}(\overline{s},\overline{\ell}{}^\forall) ~\eqdef~
\begin{cases}
\restrict_{M,p}(\overline{s},\overline{\ell}{}^\forall) & \restrict_{M,p}(\overline{s},\overline{\ell}{}^\forall) \neq \emptyset\\
\restrict_{M,\top}(\overline{s},\overline{\ell}{}^\forall)
& \text{otherwise}
\\
\end{cases}
\end{equation*}

%
%
The definition ensures that $\restricttag_{M,p}(\overline{s},\overline{\ell}{}^\forall) \neq \emptyset$  for every $(\overline{s},\overline{\ell}{}^\forall)$.

\paragraph{CHC encoding}
We adapt the formulas in \Cref{beyond:formulas:arbiter} to the case of abstract labels. 
We define $\U=\allScheds\times\ALbl$\iflong, where $\allScheds = \mathcal{P}(\{1,\ldots,k\})\setminus \{\varnothing\}$ is the set of possible schedules 
and $\ALbl$ is the set of relational predicates (restrictions).
%
We introduce an unknown predicate $A_{M,p}$ for each possible choice $\langle M, p\rangle$ of the verifier.
\else. \fi
The formulas from \Cref{beyond:formulas:arbiter} carry over, except that
\begin{inparaenum}[(i)]
\item $\U$
ranges over $\langle M,p \rangle$ instead of $\langle M, \overline{\ell}{}^\exists \rangle$, and %
\item the definition of $\delta_{M,\overline{\ell}{}^\exists}$ from the finite-branching case is  replaced with $\delta_{M,p}$, which captures the transitions according to the abstract labels.
\end{inparaenum}
%
For a schedule 
$M\in\allScheds$ and $p \in \ALbl$, 
$\delta_{M,p}(\V, \V', \LvocA)$ is defined as follows:
\[
\begin{array}{rll}
\allowed_{M,p}(\V, \LvocA) &~\eqdef~& \exists \V',\LvocE \cdot \Delta_{M}(\V, \V', \Lvoc) \wedge p(\V') \\
\delta_{M,p}(\V, \V', \LvocA) &~\eqdef~&
\big(\exists \LvocE \cdot \Delta_{M}(\V, \V', \Lvoc)\big)
\wedge 
\big(\allowed_{M,p}(\V, \LvocA) \to p(\V')\big)
\end{array}
\]
%
where $\Delta_M$ is the formula from \Cref{sec:beyond:finite} that represents the valuations $\overline{s},\overline{s}',\overline{\ell}$ s.t.\ 
the composed system 
has a transition subject to $M$ from $\overline{s}$ to $\overline{s}'$ labeled $\overline{\ell}$.
Intuitively, $\allowed_{M,p}$ captures the composed states $\overline{s}$, the  $\forall$-choices $\overline{\ell}{}^\forall$ and the schedules $M$ for which the restriction to $p$ is nonempty (i.e., \emph{some transition} is possible).
That is, $(\overline{s},\overline{\ell}{}^\forall)\models \allowed_{M,p}$ iff $\restrict_{M,p}(\overline{s},\overline{\ell}{}^\forall) \neq \emptyset$.
%
The implication in $\delta_{M,p}$ then makes sure that a vacuous restriction---i.e., one that leaves no successors---is not applied, but, instead, replaced with the full set of transitions consistent with $\overline{s},
\overline{\ell}{}^\forall, M$.
%
Namely, $\delta_{M,p}$ captures the semantics of $\restricttag_{M,p}(\overline{s},\overline{\ell}{}^\forall)$ (projected on $\overline{s},\overline{s}',\overline{\ell}{}^\forall$). 
Thus, 
in contrast to $\delta_{M,\overline{\ell}{}^\exists}$, which relates each verifier state 
$(\overline{s},\overline{\ell}{}^\forall)$ to exactly \emph{one} falsifier state $\overline{s}'$,
the abstract variant, $\delta_{M,p}$, relates $(\overline{s},\overline{\ell}{}^\forall)$ to a nonempty \emph{set} of falsifier states $\overline{s}'$, according to $\restricttag_{M,p}(\overline{s},\overline{\ell}{}^\forall)$.

Recall that we require that \emph{all} of the post-states obtained by  the verifier's choice $\langle M,p \rangle$ from  $(\overline{s},\overline{\ell}{}^\forall)$ are winning for the verifier.
In fact, the formulas in line 4 of 
\Cref{beyond:formulas:arbiter} already quantify over \emph{all} of the post-states w.r.t.\
$\delta_{M,\overline{\ell}{}^\exists}$. 
The difference is that, given $(\overline{s},\overline{\ell}{}^\forall)$, in the case of $\delta_{M,\overline{\ell}{}^\exists}$,
there is exactly one such post-state, while in the case of 
$\delta_{M,p}$, there is a set of such post-states.
Hence, once 
$\delta_{M,\overline{\ell}{}^\exists}$ is replaced by $\delta_{M,p}$, no additional modification to the encoding is needed.

The resulting encoding is sound, but, unlike the case of finite branching, not complete.

\begin{theorem}\label{thm:infinite-case}
If the set of formulas in \Cref{beyond:formulas:arbiter} adapted to $\ALbl$ 
is satisfiable, 
then $\TS \modelsg \varphi$.
\end{theorem}
\iflong
\begin{proof}
A solution for $A_{M,p}$ defines a set of moves, \emph{all} of which are guaranteed to be winning for the verifier. The definition of $\delta_{M,p}$ ensures that this set is not empty. Therefore, a strategy that chooses one of the moves in the set is winning.
\end{proof}
\fi

\begin{example}
Going back to the example in \Cref{beyond:example}, choosing schedule $M=\{2\}$ and restriction $\ell^\sharp = (\mathit{sum}_1 = \mathit{sum}_2 + A_2[i_2] + y_2)$
when $\pc_2 = 13$ ensures that the unique value of $y_2$ that satisfies the restriction is selected.
With this value chosen, the assignment of the next line will produce a value of $\mathit{sum}_2$ that is equal to that of $\mathit{sum}_1$.
This gives rise to the following winning strategy (at every iteration): (i) schedule $\{1\}$ with any restriction until $\pc_1=7$; (ii) schedule $\{2\}$ until $\pc_2=13$, then schedule $\{2\}$ again with $\ell^\sharp = (\mathit{sum}_1 = \mathit{sum}_2 + A_2[i_2] + y_2)$, then $\{2\}$ again with any restriction;
(iii) conclude the iteration by scheduling $\{1,2\}$.
As explained, the inductive invariant $\mathit{sum}_1=\mathit{sum_2}$ is preserved in this behavior,
\emph{and} there are no ``stuck'' states (since, by  construction of $\delta_{M,p}$,
empty restrictions are lifted to the full set of transitions).
\end{example}

As a corollary of \Cref{thm:infinite-case}, satisfiability of the aforementioned formulas ensures that 
$\TS \models \varphi$. 
To obtain an equi-satisfiable CHC encoding, we 
apply the transformation of \Cref{transformation}.
The resulting CHC encoding consists  of the formulas in \Cref{beyond:formulas:doomed} adapted to use $\ALbl$ in the same way the formulas in \Cref{beyond:formulas:arbiter} are adapted.

\begin{corollary}
If the set of CHCs in \Cref{beyond:formulas:doomed} adapted to $\ALbl$ 
is satisfiable, 
then $\TS \modelsg \varphi$.
\end{corollary}

%% file: 07-evaluation.tex
\section{Evaluation}
\label{eval}



\noindent
We implemented our CHC-encoding approach in a tool called \ours, 
on top of Z3~\cite{DBLP:conf/tacas/MouraB08} (4.12.0) through its Python API, using \Spacer~\cite{DBLP:conf/cav/KomuravelliGC14,DBLP:conf/cav/Gurfinkel22} as a CHC solver.
\ours takes as input a CFG, or several CFGs,
whose transitions are annotated with two-vocabulary first-order formulas, 
and constructs a formula expressing the transition relation $\Tr$. 
%
The specification is provided as: \begin{enumerate*}[label=(\roman*)]
    \item a quantifier prefix $\forall \forall$, $\forall \exists$, or $\forall \forall \exists$,
    \item observation points $\xi_i$ and 
    \item safety condition $\phi$ that must hold globally in all observations.
\end{enumerate*}
From that, the CHC encoding (\Cref{safety}, \Cref{sec:beyond}) is constructed and
passed to \Spacer for solving.
\ours supports all first-order theories supported by \Spacer (in our experiments, we used the theories of integer arithmetic and arrays).
\ours further provides the option to apply predicate abstraction with a user-provided set of predicates, same as \cite{DBLP:conf/cav/BeutnerF22}. The abstraction is incorporated into the CHC encoding using the implicit abstraction encoding~\cite{DBLP:conf/tacas/CimattiGMT14}. Notably,
many of the benchmarks shown here are solved by \ours \emph{even without} an abstraction, that is, 
directly over the concrete state.

In the area of hyperproperty verification,
there are already several tools present, and the objective
of our evaluation is to compare with such.
Still, the field is not mature enough to have a 
standardized specification format (as is the case with SMTLIB and SV-COMP, to name a few).
As a result, each tool has its own, opinionated, format, which varies from logical formulas to control-flow graphs. This makes it technically difficult to compare results of multiple solutions.
In particular, benchmarks taken from previous work come in a range of formats, dictated by the tools that introduced them.
A few of the benchmarks were translated by previous
authors and, thanks to their efforts,
are available in more than one format.
For the majority of them, 
manual work is required for translating the benchmarks, and, more importantly, there is no one accepted translation, and the translation can introduce artifacts in the evaluation.

This forced us to prioritize the comparisons in our experiments. We chose to focus on comparison with the most closely related tools to our work. These are HyPA~\cite{DBLP:conf/cav/BeutnerF22}, Pdsc~\cite{DBLP:conf/cav/ShemerGSV19}, and PCsat~\cite{DBLP:conf/cav/UnnoTK21}.
HyPA is the most recent tool, and has already collected benchmarks from various previous papers (including Weaver~\cite{DBLP:conf/cav/FarzanV19});
Pdsc and PCsat both use the same first-order encoding as our starting point and thus are also relevant.
HyPA's benchmarks include, in particular, $\forall^*\exists^*$ examples such as GNI, and Pdsc targets non-trivial alignments, and, as such all of its benchmarks have non-lockstep alignments.

\paragraph{Benchmarks}
For the evaluation of our approach we use \emph{the full sets} of benchmarks from HyPA~\cite{DBLP:conf/cav/BeutnerF22} and Pdsc~\cite{DBLP:conf/cav/ShemerGSV19}.
The benchmarks of HyPA are divided into $k$-safety benchmarks, which are adopted from~\cite{DBLP:conf/pldi/SousaD16,DBLP:conf/cav/FarzanV19,DBLP:conf/cav/ShemerGSV19,DBLP:conf/cav/UnnoTK21},
and $\forall^* \exists^*$ benchmarks,
which include refinement properties for compiler optimizations, general refinement of nondeterministic programs and generalized non-interference (GNI).
For two benchmarks, we include both a simplified version as given in~\cite{DBLP:conf/cav/BeutnerF22}, as well as the original example. 
The benchmarks of Pdsc include more non-lockstep examples, as well as all of the comparator benchmarks from~\cite{DBLP:conf/pldi/SousaD16}.
The comparator examples consist of both safe and unsafe instances.
\SH{moved ref to Weaver and rephrased (Jan 2024)}
Weaver~\cite{DBLP:conf/cav/FarzanV19} considers 12 additional (sequential) $k$-safety benchmarks. 
As an additional test case, we manually translated the running example from Weaver, which is a 3-safety property with a nontrivial alignment,
and tested it with HyHorn – HyHorn solved it in 2.25 seconds when provided with a few simple predicates (inequalities between program variables). We believe that being the running example makes it a good representative of the remaining 12. 
This brings our benchmark suite to a total of 112 $k$-safety examples (16 in \Cref{eval:table-results} plus 96 comparator benchmarks).
\SH{from reviewer: Isn't the input of "squares sum" to PCSat simplified?}

\begin{table}[t]
\newcommand\nota{\multicolumn{1}{c|}{\cancel{~~\rule{0pt}{6pt}}}}
%
\resizebox{\textwidth}{!}{
\begin{tabular}{@{\!\!}l@{}l@{}}
\begin{minipage}[t]{10cm}
 \strut\vspace*{-\baselineskip}\newline
\input{data/table-safety.tex}
\end{minipage}
&
\begin{minipage}[t]{5.7cm}
 \strut\vspace*{-\baselineskip}\newline
\input{data/table-beyond.tex}
\end{minipage}
\end{tabular}
}

\vspace{3pt}
\renewcommand\nota{{/\hspace{0.5pt}}}
\caption{
Experimental results for $k$-safety properties. Time is measured in seconds. ``---'' represents timeouts after 20 minutes. ``\nota'' denotes benchmarks not present in the respective tool's suite.\\
In benchmark names, [FV19] refers to~\cite{DBLP:conf/cav/FarzanV19};
[BF22] refers to~\cite{DBLP:conf/cav/BeutnerF22}.}
\label{eval:table-results}
\vspace{-1.5em}
\end{table}

\paragraph{Experiments}

To demonstrate the effectiveness of \ours 
we compare to HyPA~\cite{DBLP:conf/cav/BeutnerF22}, the most recent approach of formal verification of $\forall^* \exists^*$-hyperproperties, which employs a construction using automata. 
To exhibit the benefits of the direct CHC encoding we also compare the $k$-safety examples to 
PCSat~\cite{DBLP:conf/cav/UnnoTK21} and \textsc{Pdsc}~\cite{DBLP:conf/cav/ShemerGSV19}. Both encode the $k$-safety problem using FOL formulas as in \Cref{safety:formulas:arbiter}. PCSat uses a specialized solver for pfwCSP (a fragment of FOL that includes these formulas), while \textsc{Pdsc}
solves the FOL formulas by enumerating alignments and using a CHC solver for each alignment.
%
%
We do not compare to game solvers since, as reported by~\cite{DBLP:conf/cav/BeutnerF22}, state-of-the-art infinite-state game solvers, such as~\cite{DBLP:journals/pacmpl/FarzanK18,DBLP:conf/cav/BaierCFFJS21}, which work without user-provided predicates, are unable to solve the benchmarks we consider.

We run \ours on the full set of benchmarks, and each of the other tools on the ones 
included in their benchmark suite. This is because each tool has its own input format: HyPA and \textsc{Pdsc} each has its own representation for the transition system and the property; 
PCSat accepts 
pfwCSP instances that are constructed manually. 
Some of the benchmarks are common to several tools. 




All experiments are run on an AMD EPYC 74F3 with 32GB of memory.
HyPA and PCSat are executed in Docker using their published artifacts.

\paragraph{Results}
The performance measurements of the tools  for the $k$-safety benchmarks and for the $\forall^*\exists^*$ benchmarks are shown in \Cref{eval:table-results}. The results for the comparator examples are deferred to the extended version of the paper~\cite{itzhaky2023hyperproperty}. 
\ours is tested in two modes:
with predicate abstraction (``PA'') and without  (``concrete'').
HyPA and \textsc{Pdsc} require predefined predicates (the same predicates are used in all tools), while PCSat does not, but uses hints to solve `array insert' and `squares sum'.
\ours solves almost all of the benchmarks with PA in under a second, outperforming previous approaches by up to two orders of magnitude;
and also solves most of the benchmarks quickly without PA, esp. the $\forall^*\exists^*$ properties. 
In particular, \ours solves the two array benchmarks, while  HyPA and PCSat do not support arrays and only solve simplified versions with integers. 
The runtime of \ours (both with and without predicates) on the comparator examples is similar to the runtime of \textsc{Pdsc} (see ~\cite{itzhaky2023hyperproperty}\SH{add ref}\yv{this one?}), 
where \ours solves some benchmarks that \textsc{Pdsc} does not. (The other tools do not include these benchmarks.) On the unsafe examples, \ours provides a concrete counterexample, while \textsc{Pdsc} is only able to determine that there is no inductive invariant and alignment expressible with the given set of predicates.

%% file: data/table-safety.tex
\begin{tabular}{|l|rr@{~}|r@{~}|r@{~}|r@{~}|}
\hline
  $k$-safety
  & \multicolumn{2}{c|}{\ours}  & \multicolumn{1}{c|}{HyPA} & \multicolumn{1}{c|}{PCSat} & \multicolumn{1}{c|}{\textsc{Pdsc}} \\
  & ~PA~ & \multicolumn{1}{c|}{\fontsize{8pt}{8pt}\selectfont concrete} &  &  &  \\
\hline
double square NI    & 0.56 & ---  & 67.0 & ---~ & 6.8 \\
double square NI ff & 0.12 & --- &  5.3  & 1.5 & \nota \\
half square NI      & 0.30 & 0.30 & 63.0 & 13.4 & 3.4 \\
squares sum         & 0.17 & 3.41 & 70.4 & 360.7  & 2.8 \\
squares sum (simplified)
                    & 0.10  & 0.30 & 17.2 & \nota  & \nota \\
array insert        & 0.86  & 13.4 & \nota & \nota & 18.5  \\
array insert (simplified) 
                    & 1.33  & 2.58 & 16.2 & 378.6  & \nota  \\
exp1x3              & 0.08  & 0.09 & 2.9  & \nota & \nota \\
fig 3~[FV19]
                    & 0.03  &  --- & 7.9 & \nota & \nota \\
fig 2~[BF22]
                    & 0.11  &  --- & 13.6  & \nota & \nota \\
col item symm       & 0.49  & 0.49 & 14.9 & \nota & \nota \\
counter det         & 0.46  &  --- & 6.2  & \nota & \nota \\
mult equiv          & 0.29  &  --- & 14.2  & \nota & \nota \\
mult equiv (simplified)
                    & 0.19  &  --- &  10.3   &  \nota  & \nota \\
array int mod       & 0.13  &  --- &  \nota  &  \nota  & 58.2 \\
mult dist~[FV19]    & 2.25  &  --- &  \nota  &  \nota  & \nota \\
\hline
\end{tabular}

%% file: data/table-beyond.tex
\renewcommand\cheat{}  
\begin{tabular}{|l|rr@{~}|r@{~}|}
\hline
  $\forall^*\exists^*$
  & \multicolumn{2}{c|}{\ours}  & \multicolumn{1}{c|}{HyPA} \\
  & ~PA~ & \multicolumn{1}{c|}{\fontsize{8pt}{8pt}\selectfont concrete} &  \\
\hline
non-det add            & 1.45 & 2.80 & 3.3 \\
counter sum            & 0.09 & --- & 4.0 \\
async GNI              & 0.36 & 0.37 & \cheat3.8 \\
compiler opt 1         & 0.14 & 0.19 & 1.8 \\
compiler opt 2         & 0.17 & 0.78 & \cheat2.0 \\
refine                 & 0.18  & 0.29 & \cheat4.0 \\
refine 2               & 0.28  & 0.65 & 3.9 \\
smaller                & 0.16  & 0.96 & \cheat2.0 \\
counter diff           & 0.17  & ---  & 6.8 \\
fig 3~[BF22]
                       & 0.81  & ---  & 9.9 \\
P1 (simple)            & 0.19  & 0.59 & 1.4 \\
P1 (GNI)               & 0.26  & 0.75 & 138.7 \\
P2 (GNI)               & 8.50  & 6.65 & 12.8 \\
P3 (GNI)               & 0.32  & 0.20 & 4.6 \\
P4 (GNI)               & 0.77  & 0.63 & 27.7 \\
\hline
\end{tabular}

%% file: 08-related.tex
\section{Related Work}
\label{sec:related}

There is a large body of work studying verification of hyperproperties.
While earlier verification
techniques mostly focus on $k$-safety properties, or specific examples such as program equivalence, monotonocity, determinism~\cite{DBLP:conf/csfw/BartheDR04,DBLP:conf/sas/TerauchiA05,DBLP:conf/fm/BartheCK11,DBLP:journals/stvr/GodlinS13,DBLP:conf/pldi/SousaD16,DBLP:conf/cav/YangVSGM18,DBLP:conf/esop/EilersMH18,DBLP:conf/cav/FarzanV19,DBLP:conf/cav/ShemerGSV19,ACM:conf/popl/Antonopoulos23},
lately verification of non-safety
hyperproperties has been studied~\cite{DBLP:conf/lfcs/BartheCK13,DBLP:conf/cav/CoenenFST19,DBLP:conf/cav/UnnoTK21,DBLP:conf/csfw/BeutnerF22,DBLP:conf/cav/BeutnerF22}.
Below we discuss the works closest to ours.

\paragraph{$k$-Safety}
Automatic verification of $k$-safety properties can be achieved by reducing the problem to
a standard safety verification problem by means of self-composition~\cite{DBLP:conf/csfw/BartheDR04},
product-programs~\cite{DBLP:conf/fm/BartheCK11}, and their derivatives~\cite{DBLP:conf/cav/YangVSGM18,DBLP:conf/esop/EilersMH18}.
Recently, however, it was
identified that the alignment of the different copies has a substantial effect over
the complexity of the verification problem~\cite{DBLP:conf/cav/ShemerGSV19,DBLP:conf/cav/FarzanV19,DBLP:conf/pldi/ChurchillP0A19}.
Our approach is most related to the technique of Shemer \etal~\cite{DBLP:conf/cav/ShemerGSV19},
which uses a semantic alignment that chooses which copy of the system performs a move
based on the composed state of the different copies.
They suggest an
algorithm that iterates through the set of possible semantic alignments, such
that in each iteration a CHC solver tries to prove the property, with the chosen alignmnet, using predicate
abstraction. Unlike~\cite{DBLP:conf/cav/ShemerGSV19}, \ours delegates the search for the alignment
to the CHC solver, together with the search for the invariant, making the algorithm
less dependent on predicate abstraction. Moreover, while~\cite{DBLP:conf/cav/ShemerGSV19} is restricted
to $k$-safety only, 
our technique 
can handle $k$-safety as well as the more general \aehyper.

\paragraph{$\forall^* \exists^*$ Hyperproprties}
Recently, verification of $\forall^*\exists^*$ hyperproperties has been studied,
targeting both finite and infinite systems~\cite{DBLP:conf/cav/UnnoTK21,DBLP:conf/cav/CoenenFST19,DBLP:conf/cav/BeutnerF22}.
Unno \etal~\cite{DBLP:conf/cav/UnnoTK21}
present an approach based on an encoding of hyperproperties verification as
satisfiability of formulas in FOL that extend Horn form with disjunctions,
existential quantification and well founded relations.
Deciding satisfiability
of the generated set of formulas is based on a variant of the CEGIS framework.
\ours is different as it encodes \aehyper verification as a set
of CHCs, which does not require a specialized solver and can use
any off-the-shelf CHC solver.
Coenen et.al.~\cite{DBLP:conf/cav/CoenenFST19} suggested a game-based approach for verification of $\forall^*\exists^*$ properties over finite-state systems, which was then extended by
Beutner \etal~\cite{DBLP:conf/cav/BeutnerF22} to handle infinite-state systems.
Similarly to~\cite{DBLP:conf/cav/BeutnerF22}, we use game semantics to
solve $\forall^*\exists^*$ problems, but do not
require building the game-graph in order to solve the game, instead reducing
the game solution to satisfiability of CHCs.
It is important
to note that in the case of infinite branching degree, while the
approach in~\cite{DBLP:conf/cav/BeutnerF22}
explicitly checks for emptiness of restrictions in hindsight, i.e., after they are used in a strategy, and
removes them iteratively if needed, \ours embeds
the emptiness requirements into the set of CHCs. 
Recently, \citet{DBLP:conf/csfw/BeutnerF22} extended the game-based approach
to use prophecy variables as a way to achieve completeness of the reduction to games.
Extending our approach to this case is a promising avenue for future research.

\paragraph{Relational CHCs}
\citet{DBLP:conf/fmcad/MordvinovF19}
present a method for discovering relational solutions to CHCs.
Their setting is different: the inputs are CHCs that serve as the definition of the transitions,
and synchronization is between sets of unknown predicates; at the current state, only lock-step semantics is considered.
Furthermore, their algorithm extends and modifies \Spacer~\cite{DBLP:conf/cav/KomuravelliGC14}, while our approach can use any CHC solver without modification.

\paragraph{Infinite-State Game Solving}
%
Our approach for verifying $\forall^* \exists^*$ hyperproperties
is based on the game semantics of \aehyper proposed in~\cite{DBLP:conf/cav/CoenenFST19,DBLP:conf/cav/BeutnerF22}. 
However, we do not propose a general game solving algorithm. Instead, 
we use the game semantics to come up with a first-order encoding of hyperproperty verification problems,
which is then reduced to CHC solving. This allows us to use any CHC solver when solving the hyperproperty game.
There is a large body of work on solving infinite-state games~\cite{DBLP:conf/concur/AlfaroHM01,DBLP:conf/popl/BeyeneCPR14,DBLP:conf/fmcad/WalkerR14,DBLP:journals/pacmpl/FarzanK18}.  
%
The game solving approach in ~\cite{DBLP:conf/fmcad/WalkerR14} uses three-valued predicate abstraction to reduce the problem to finite-state game solving and requires to iteratively refine the controllable predecessor operator when computing candidate winning states. The approach in~\cite{DBLP:journals/pacmpl/FarzanK18} targets games defined over the theory of linear real arithmetic and is based on an unrolling of the game and the use of Craig interpolants~\cite{craig57} to synthesize a winning strategy. The game solver in~\cite{DBLP:conf/cav/BaierCFFJS21} is not restricted to a given FOL theory, but requires an interpolation procedure in order to compute sub-goals that are used to inductively split a game into sub-games.
As reported by~\cite{DBLP:conf/cav/BeutnerF22}, game solving approaches~\cite{DBLP:journals/pacmpl/FarzanK18,DBLP:conf/cav/BaierCFFJS21},
which work without a provided set of predicates, are unable to handle the infinite-state games for the benchmarks we consider.  Moreover, the approaches in~\cite{DBLP:journals/pacmpl/FarzanK18,DBLP:conf/cav/CoenenFST19,DBLP:conf/cav/BaierCFFJS21,DBLP:conf/cav/BeutnerF22} cannot handle games that are defined using formulas over the theory of arrays, which are part of our benchmark. 
%
%
%
The approach of~\citet{DBLP:conf/popl/BeyeneCPR14} to solving games over infinite graphs
is based on reduction of games (including safety games) to CHCs. However,
unlike the reduction presented in this paper, in~\cite{DBLP:conf/popl/BeyeneCPR14}
the games are encoded in a different fragment of Horn, namely $\forall\exists$-Horn
where the head predicates can contain existential quantifiers.
%
More recently (and concurrently with our work),~\cite{Faella_Parlato_2023} proposed a new reduction of game solving to CHC solving.
Their approach handles safety games in which the branching degree of the ``safe'' player (the verifier in our setting) is bounded.
In contrast, our encoding supports also infinite branching with the restrictions mechanism.
Moreover, they do not support predicate abstraction, which is crucial for solving some of our benchmarks.
\paragraph{Restrictions as Underapproximations} The use of restrictions as underapproxmations of the transition relation, inspired by~\citet{DBLP:conf/cav/BeutnerF22}, corresponds to the use of must hyper-transitions~\cite{DBLP:conf/lics/LarsenX90} in abstract transition systems~\cite{DBLP:conf/tacas/ShohamG04,DBLP:conf/lics/DamsN04} and games~\cite{DBLP:conf/lics/AlfaroGJ04,DBLP:conf/concur/AlfaroR07}.
Similarly to~\citet{DBLP:conf/popl/GodefroidNRT10,
acm/CookK2013}, we use such underapproximations to replace an existential quantifier by universal quantification \emph{within} the restriction.


%% file: 09-conclusion.tex
\section{Conclusion}
\label{sec:conc}

We introduced a translation of a family of non-Horn first-order formulas to CHCs.
This translation led to the first CHC encoding of a simultaneous inference of  an invariant and an alignment for verifying $k$-safety properties.
While the transformation itself is rather simple, identifying it was not straightforward and alluded previous works on the topic.
We have further extended the CHC encoding to infer a witness function for existentially quantified traces  arising in the verification of  \aehyper properties. 
Our experiments exhibit significant improvement over state-of-the-art hyperproperty verifiers thanks to the existence of advanced off-the-shelf CHC solvers,
whose efficacy is expected to improve even further. The approach shows promising capabilities in solving (many) hyeprproperty verification problems completely automatically. In some cases, predicates still have to be provided by the user, a limitation that we hope to overcome in the future by automatic inference of predicates. 
Applying (or extending) the transformation to obtain CHC encoding for other verification fragments is an interesting direction for future work.

%% file: B-appendix-comparators.tex
\section{Comprehensive Comparison of the Comparators}
\label{annex:comparators}

The comparators benchmark suite from \citet{DBLP:conf/pldi/SousaD16} contains
a set of comparator functions used in data ordering.
Each comparator function is expected to satisfy the following three properties:
\begin{enumerate}
    \item[(P1)] asymmetry: $\forall x, y.~ \lnot(x < y \land x > y)$
    \item[(P2)] transitivity: $\forall x, y, z.~ (x \leq y \land y \leq z) \rightarrow x \leq z$
    \item[(P3)] congruence: $\forall x, y, z.~ x = y \rightarrow (x < z \leftrightarrow y < z)$
\end{enumerate}

The functions accept as input objects to compare, with may contain integers and arrays.
Per convention, the return value is $-1$ if $x < y$, $0$ if $x = y$, and $1$ if $x > y$.

These benchmarks were evaluated in \citet{DBLP:conf/cav/ShemerGSV19} with their tool, \textsc{Pdsc}, where performance comparable to \textsc{Synonym}~\cite{DBLP:conf/cav/PickFG18} was demonstrated.
We ran the same set of benchmarks and compared the running time of \ours with that of \textsc{Pdsc}.
Like in \Cref{eval}, we tested \ours in two modes, with predicate abstraction and without it.
In both cases, \textsc{Pdsc} results represents execution with predicates, because the latter does not have a concrete solving mode.

Results are shown in \Cref{annex:hyhorn-pdsc-plot}.
Preprocessing time is omitted from the comparison and only solver time of both tools is compared.
\ours outperforms \textsc{Pdsc} on most benchmarks when using predicate abstraction.
Without predicates, \ours still outperforms on some cases, and is not much slower than \textsc{Pdsc} on others even though \textsc{Pdsc} is using the abstraction.

Furthermore, \textsc{Pdsc} and \ours with PA verified 54 instances in total (the same instances), whereas the concrete mode of \ours verified 61 instances.

Raw data is also included in the table on the next page.
Correct comparator functions are labeled with ``-true'', buggy versions are labeled with ``-false''.
The ``out'' column tells whether the tool was able to verify the instance.
``$i$'' is the number of SMT queries used by \textsc{Pdsc} (\ours always uses exactly one),
and ``\#$p$'' is the number of predicate used for the abstraction.

\begin{figure}[b]
    \input{gfx/scatter-plot}
    \caption{Solver time for \ours vs. \textsc{Pdsc} on the comparator benchmarks (\cite{DBLP:conf/pldi/SousaD16}).
    Time is in seconds, scale is logarithmic.
    \ours was run in two modes: with predicate abstraction ("PA") and without ("Concrete").
    \textsc{Pdsc} was run with PA in both cases because this is the only mode it supports.}
    \label{annex:hyhorn-pdsc-plot}
\end{figure}

\begin{table}
\begin{tabular}{@{}ll@{}}
\scalebox{0.5}{
\begin{minipage}{0.9\textwidth}
\include{data/table-comparators-part1}
\end{minipage}
}
&
\scalebox{0.5}{
\begin{minipage}{0.9\textwidth}
\include{data/table-comparators-part2}
\end{minipage}
}
\end{tabular}
\end{table}

%% file: gfx/scatter-plot.tex
\newcommand\tthesy{\textrm{T}}
\newcommand\thipster{\textrm{H}}

\newcommand\fwith{w\hspace{-1pt}/\hspace{-1pt}\xspace}
\newcommand\fwithout{w\hspace{-1pt}/\hspace{-1pt}o\xspace}

\begin{tikzpicture}
\coordinate (O) at (0,0);
\begin{axis}[
    title={PA},
    at=(O),
    draw={black!30!white},
    width=5.5cm, height=5.5cm,
    grid=both,
    xmode=log, ymode=log,
    xmin=0.05, xmax=10,
    ymin=0.05, ymax=10,
    ticklabel style={font=\tiny},
    x label style={at={(axis description cs:0.5,-0.07)},anchor=north},
    y label style={at={(axis description cs:-0.15,.5)},anchor=south},
    xlabel={\textsc{Pdsc}},
    ylabel={\ours}
]
\addplot[  
    only marks,
    scatter,
    scatter src=(y-x)/4,
    point meta min=-.6, point meta max=1,
    mark size=1.25pt]
table[x=solver,y=hyhorn]{gfx/scatter-plot.dat.txt};
\end{axis}
\draw[color=olive!60!white] (O) -- +(4cm, 4cm);


\coordinate (O) at (6.3cm,0);
\begin{axis}[
    title={Concrete},
    at=(O),
    draw={black!30!white},
    width=5.5cm, height=5.5cm,
    grid=both,
    xmode=log, ymode=log,
    xmin=0.05, xmax=10,
    ymin=0.05, ymax=10,
    ticklabel style={font=\tiny},
    x label style={at={(axis description cs:0.5,-0.07)},anchor=north},
    y label style={at={(axis description cs:-0.15,.5)},anchor=south},
    xlabel={\hspace{0.5cm}\textsc{Pdsc} {\fontsize{7pt}{7pt}\selectfont (w/ PA)}},
    ylabel={\ours}
]
\addplot[  
    only marks,
    scatter,
    scatter src=(y-x)/4,
    point meta min=-.5, point meta max=1,
    mark size=1.25pt]
table[x=solver,y=hyhorn-pa]{gfx/scatter-plot.dat.txt};
\end{axis}
\draw[color=olive!60!white] (O) -- +(4cm, 4cm);
\end{tikzpicture}

%% file: data/table-comparators-part1.tex
\begin{tabular}{|lc|rrcrr|rcrc|}
\hline
 & & \multicolumn{5}{c|}{\textsc{Pdsc} (PA)} & \multicolumn{4}{c|}{\ours} \\
 & & \multicolumn{2}{c|}{time} & & & & \multicolumn{2}{c|}{PA} & \multicolumn{2}{c|}{Concrete} \\
Name & prop & \multicolumn{1}{c}{p.p.} & \multicolumn{1}{c|}{solver} & out & $i$ & \#$p$ & solver & \multicolumn{1}{c|}{out} & solver & out \\
\hline
\input{data/table-comparators-part1.dat}%
\\ \hline
\end{tabular}

%% file: data/table-comparators-part1.dat.tex
ArrayInt-false  &  1  &  0.072  &  0.263  &  Y  &  1  &  8  &  0.158  &  Y  &  0.280  &  Y   \\
ArrayInt-false  &  2  &  0.261  &  0.488  &  Y  &  1  &  14  &  0.361  &  Y  &  0.484  &  Y   \\
ArrayInt-false  &  3  &  0.070  &  0.301  &  N  &  2  &  17  &  0.773  &  N  &  1.430  &  N   \\
ArrayInt-true  &  1  &  0.077  &  0.280  &  Y  &  1  &  8  &  0.065  &  Y  &  0.316  &  Y   \\
ArrayInt-true  &  2  &  0.267  &  0.504  &  Y  &  1  &  14  &  0.374  &  Y  &  0.594  &  Y   \\
ArrayInt-true  &  3  &  0.332  &  0.579  &  Y  &  1  &  17  &  0.380  &  Y  &  0.485  &  Y   \\
CatBPos-false  &  1  &  0.021  &  0.145  &  N  &  1  &  11  &  0.044  &  N  &  0.092  &  N   \\
CatBPos-false  &  2  &  0.039  &  0.202  &  N  &  2  &  20  &  0.153  &  N  &  0.232  &  N   \\
CatBPos-false  &  3  &  0.042  &  0.214  &  N  &  2  &  23  &  0.185  &  N  &  0.234  &  N   \\
Chromosome-false  &  1  &  0.040  &  0.903  &  N  &  1  &  7  &  0.067  &  N  &  0.612  &  Y   \\
Chromosome-false  &  2  &  0.074  &  1.018  &  N  &  2  &  12  &  0.257  &  N  &  0.926  &  N   \\
Chromosome-false  &  3  &  0.124  &  1.072  &  N  &  2  &  15  &  0.261  &  N  &  0.942  &  N   \\
Chromosome-true  &  1  &  0.022  &  0.236  &  N  &  1  &  8  &  0.042  &  N  &  0.139  &  Y   \\
Chromosome-true  &  2  &  0.061  &  0.315  &  N  &  2  &  14  &  0.397  &  N  &  0.324  &  Y   \\
Chromosome-true  &  3  &  0.395  &  0.649  &  Y  &  1  &  17  &  0.321  &  Y  &  0.241  &  Y   \\
ColItem-false  &  1  &  0.021  &  0.110  &  N  &  1  &  9  &  0.042  &  N  &  0.061  &  N   \\
ColItem-false  &  2  &  0.040  &  0.158  &  N  &  2  &  15  &  0.136  &  N  &  0.228  &  N   \\
ColItem-false  &  3  &  0.081  &  0.197  &  N  &  2  &  18  &  0.169  &  N  &  0.218  &  N   \\
ColItem-true  &  1  &  0.045  &  0.148  &  Y  &  1  &  9  &  0.045  &  Y  &  0.148  &  Y   \\
ColItem-true  &  2  &  0.076  &  0.205  &  Y  &  1  &  15  &  0.106  &  Y  &  0.246  &  Y   \\
ColItem-true  &  3  &  0.120  &  0.255  &  Y  &  1  &  18  &  0.118  &  Y  &  0.215  &  Y   \\
Contact-false  &  1  &  0.057  &  0.162  &  Y  &  1  &  7  &  0.049  &  Y  &  1.271  &  Y   \\
Contact-false  &  2  &  0.046  &  0.181  &  N  &  2  &  12  &  0.146  &  N  &  0.366  &  N   \\
Contact-false  &  3  &  0.048  &  0.190  &  N  &  2  &  15  &  0.169  &  N  &  0.322  &  N   \\
Container-false-v1  &  1  &  0.016  &  0.087  &  N  &  1  &  9  &  0.038  &  N  &  0.055  &  N   \\
Container-false-v1  &  2  &  0.032  &  0.122  &  N  &  2  &  15  &  0.128  &  N  &  0.143  &  N   \\
Container-false-v1  &  3  &  0.095  &  0.190  &  Y  &  1  &  18  &  0.102  &  Y  &  0.115  &  Y   \\
Container-false-v2  &  1  &  0.017  &  0.098  &  N  &  1  &  9  &  0.039  &  N  &  0.043  &  N   \\
Container-false-v2  &  2  &  0.032  &  0.138  &  N  &  2  &  15  &  0.134  &  N  &  0.171  &  N   \\
Container-false-v2  &  3  &  0.063  &  0.176  &  N  &  2  &  18  &  0.160  &  N  &  0.182  &  N   \\
Container-true  &  1  &  0.050  &  0.139  &  Y  &  1  &  9  &  0.045  &  Y  &  0.124  &  Y   \\
Container-true  &  2  &  0.082  &  0.197  &  Y  &  1  &  15  &  0.103  &  Y  &  0.278  &  Y   \\
Container-true  &  3  &  0.122  &  0.241  &  Y  &  1  &  18  &  0.106  &  Y  &  0.221  &  Y   \\
DataPoint-false  &  1  &  0.023  &  0.128  &  N  &  1  &  7  &  0.046  &  N  &  0.071  &  N   \\
DataPoint-false  &  2  &  0.043  &  0.168  &  N  &  2  &  12  &  0.139  &  N  &  0.292  &  N   \\
DataPoint-false  &  3  &  0.043  &  0.172  &  N  &  2  &  15  &  0.153  &  N  &  0.397  &  N   \\
FileItem-false  &  1  &  0.034  &  0.080  &  Y  &  1  &  5  &  0.033  &  Y  &  0.078  &  Y   \\
FileItem-false  &  2  &  0.047  &  0.105  &  Y  &  1  &  9  &  0.062  &  Y  &  0.083  &  Y   \\
FileItem-false  &  3  &  0.031  &  0.096  &  N  &  2  &  12  &  0.117  &  N  &  0.135  &  N   \\
FileItem-true  &  1  &  0.038  &  0.101  &  Y  &  1  &  5  &  0.037  &  Y  &  0.102  &  Y   \\
FileItem-true  &  2  &  0.061  &  0.149  &  Y  &  1  &  9  &  0.074  &  Y  &  0.151  &  Y   \\
FileItem-true  &  3  &  0.090  &  0.180  &  Y  &  1  &  12  &  0.081  &  Y  &  0.171  &  Y   \\
IsoSprite-false-v1  &  1  &  0.013  &  0.064  &  N  &  1  &  1  &  0.031  &  N  &  0.037  &  N   \\
IsoSprite-false-v1  &  2  &  0.026  &  0.097  &  N  &  2  &  9  &  0.105  &  N  &  0.147  &  N   \\
IsoSprite-false-v1  &  3  &  0.029  &  0.105  &  N  &  2  &  12  &  0.125  &  N  &  0.142  &  N   \\
IsoSprite-false-v2  &  1  &  0.027  &  0.200  &  N  &  1  &  13  &  0.051  &  N  &  0.066  &  N   \\
IsoSprite-false-v2  &  2  &  0.053  &  0.276  &  N  &  2  &  21  &  0.192  &  N  &  0.313  &  N   \\
IsoSprite-false-v2  &  3  &  0.175  &  0.399  &  Y  &  1  &  24  &  0.145  &  Y  &  0.235  &  Y   

%% file: data/table-comparators-part2.tex
\begin{tabular}{|lc|rrcrr|rcrc|}
\hline
 & & \multicolumn{5}{c|}{\textsc{Pdsc} (PA)} & \multicolumn{4}{c|}{\ours} \\
 & & \multicolumn{2}{c|}{time} & & & & \multicolumn{2}{c|}{PA} & \multicolumn{2}{c|}{Concrete} \\
Name & prop & \multicolumn{1}{c}{p.p.} & \multicolumn{1}{c|}{solver} & out & $i$ & \#$p$ & solver & \multicolumn{1}{c|}{out} & solver & out \\
\hline
\input{data/table-comparators-part2.dat}
\\
\hline
\end{tabular}

%% file: data/table-comparators-part2.dat.tex
Match-false  &  1  &  0.017  &  0.092  &  N  &  1  &  7  &  0.040  &  N  &  0.046  &  N   \\
Match-false  &  2  &  0.048  &  0.145  &  Y  &  1  &  12  &  0.075  &  Y  &  0.093  &  Y   \\
Match-false  &  3  &  0.034  &  0.134  &  N  &  2  &  15  &  0.157  &  N  &  0.204  &  N   \\
Match-true  &  1  &  0.042  &  0.116  &  Y  &  1  &  7  &  0.039  &  Y  &  0.088  &  Y   \\
Match-true  &  2  &  0.064  &  0.154  &  Y  &  1  &  12  &  0.081  &  Y  &  0.112  &  Y   \\
Match-true  &  3  &  0.097  &  0.191  &  Y  &  1  &  15  &  0.086  &  Y  &  0.140  &  Y   \\
NameComparator-false  &  1  &  0.024  &  0.154  &  N  &  1  &  6  &  0.079  &  N  &  0.142  &  N   \\
NameComparator-false  &  2  &  0.219  &  0.379  &  Y  &  1  &  12  &  0.270  &  Y  &  0.537  &  Y   \\
NameComparator-false  &  3  &  0.368  &  0.539  &  Y  &  1  &  15  &  0.254  &  Y  &  0.553  &  Y   \\
NameComparator-true  &  1  &  0.075  &  0.241  &  Y  &  1  &  8  &  0.060  &  Y  &  0.455  &  Y   \\
NameComparator-true  &  2  &  0.351  &  0.556  &  Y  &  1  &  15  &  0.481  &  Y  &  1.266  &  Y   \\
NameComparator-true  &  3  &  0.194  &  0.403  &  Y  &  1  &  18  &  0.429  &  Y  &  0.693  &  Y   \\
Node-false  &  1  &  0.068  &  0.259  &  Y  &  1  &  11  &  0.053  &  Y  &  0.251  &  Y   \\
Node-false  &  2  &  0.114  &  0.333  &  Y  &  1  &  18  &  0.123  &  Y  &  0.159  &  Y   \\
Node-false  &  3  &  0.083  &  0.309  &  N  &  2  &  21  &  0.216  &  N  &  0.490  &  N   \\
Node-true  &  1  &  0.066  &  0.295  &  Y  &  1  &  11  &  0.051  &  Y  &  0.204  &  Y   \\
Node-true  &  2  &  0.097  &  0.361  &  Y  &  1  &  18  &  0.124  &  Y  &  0.421  &  Y   \\
Node-true  &  3  &  0.154  &  0.422  &  Y  &  1  &  21  &  0.129  &  Y  &  0.882  &  Y   \\
NzbFile-false  &  1  &  0.038  &  0.523  &  N  &  1  &  9  &  0.060  &  N  &  0.483  &  N   \\
NzbFile-false  &  2  &  0.141  &  0.688  &  Y  &  1  &  15  &  0.188  &  Y  &  1.161  &  Y   \\
NzbFile-false  &  3  &  0.187  &  0.732  &  Y  &  1  &  18  &  0.168  &  Y  &  1.118  &  Y   \\
NzbFile-true  &  1  &  0.075  &  0.498  &  N  &  1  &  14  &  0.115  &  N  &  0.443  &  Y   \\
NzbFile-true  &  2  &  0.174  &  0.678  &  N  &  2  &  23  &  0.477  &  N  &  1.827  &  Y   \\
NzbFile-true  &  3  &  0.154  &  0.656  &  N  &  2  &  26  &  0.556  &  N  &  3.002  &  Y   \\
Solution-false  &  1  &  0.061  &  0.178  &  Y  &  1  &  7  &  0.047  &  Y  &  0.732  &  Y   \\
Solution-false  &  2  &  0.096  &  0.227  &  Y  &  1  &  13  &  0.115  &  Y  &  1.559  &  Y   \\
Solution-false  &  3  &  0.093  &  0.231  &  N  &  2  &  16  &  0.166  &  N  &  0.591  &  N   \\
Solution-true  &  1  &  0.054  &  0.162  &  Y  &  1  &  7  &  0.054  &  Y  &  1.197  &  Y   \\
Solution-true  &  2  &  0.092  &  0.225  &  Y  &  1  &  13  &  0.131  &  Y  &  2.974  &  Y   \\
Solution-true  &  3  &  0.134  &  0.274  &  Y  &  1  &  16  &  0.205  &  Y  &  1.034  &  Y   \\
TextPosition-false  &  1  &  0.050  &  0.164  &  Y  &  1  &  9  &  0.049  &  Y  &  0.222  &  Y   \\
TextPosition-false  &  2  &  0.044  &  0.190  &  N  &  2  &  15  &  0.148  &  N  &  0.345  &  N   \\
TextPosition-false  &  3  &  0.044  &  0.197  &  N  &  2  &  18  &  0.178  &  N  &  0.366  &  N   \\
TextPosition-true  &  1  &  0.069  &  0.203  &  Y  &  1  &  11  &  0.053  &  Y  &  0.350  &  Y   \\
TextPosition-true  &  2  &  0.108  &  0.281  &  Y  &  1  &  18  &  0.135  &  Y  &  0.462  &  Y   \\
TextPosition-true  &  3  &  0.147  &  0.327  &  Y  &  1  &  21  &  0.148  &  Y  &  0.418  &  Y   \\
Time-false  &  1  &  0.013  &  0.048  &  N  &  1  &  5  &  0.032  &  N  &  0.048  &  N   \\
Time-false  &  2  &  0.042  &  0.088  &  Y  &  1  &  9  &  0.058  &  Y  &  0.083  &  Y   \\
Time-false  &  3  &  0.021  &  0.070  &  Y  &  1  &  12  &  0.049  &  Y  &  0.042  &  Y   \\
Time-true  &  1  &  0.031  &  0.072  &  Y  &  1  &  5  &  0.032  &  Y  &  0.053  &  Y   \\
Time-true  &  2  &  0.041  &  0.089  &  Y  &  1  &  9  &  0.060  &  Y  &  0.107  &  Y   \\
Time-true  &  3  &  0.066  &  0.115  &  Y  &  1  &  12  &  0.065  &  Y  &  0.095  &  Y   \\
Word-false  &  1  &  0.050  &  0.495  &  N  &  1  &  11  &  0.102  &  N  &  0.381  &  N   \\
Word-false  &  2  &  0.175  &  0.691  &  N  &  2  &  19  &  0.673  &  N  &  1.360  &  Y   \\
Word-false  &  3  &  0.046  &  0.570  &  Y  &  1  &  22  &  0.201  &  Y  &  0.575  &  Y   \\
Word-true  &  1  &  0.079  &  0.387  &  Y  &  1  &  10  &  0.066  &  Y  &  0.374  &  Y   \\
Word-true  &  2  &  0.183  &  0.545  &  Y  &  1  &  17  &  0.326  &  Y  &  1.096  &  Y   \\
Word-true  &  3  &  0.415  &  0.779  &  Y  &  1  &  20  &  0.317  &  Y  &  0.807  &  Y   